%% file: main.tex
\documentclass[11pt]{article}
\usepackage[margin=1in]{geometry}
\usepackage{rpmacros}
\usepackage[T1]{fontenc}
\usepackage[usenames,dvipsnames]{color}
\usepackage{stmaryrd}
\usepackage{comment}

\usepackage{lineno}
\usepackage{xfrac}
\usepackage{mathpazo}
\usepackage{todonotes}
\usepackage{setspace}
\usepackage{nicefrac}
\usepackage{gitinfo}
\usepackage{float}
\usepackage{scrextend}
\usepackage{tikz}
\allowdisplaybreaks

\usepackage{caption}
\usepackage{blindtext}

\usepackage{bm}
\usepackage{blkarray}
\usepackage{amsmath}
\usepackage{multirow}
\usepackage{fullpage}
\usepackage[utf8]{inputenc}
\usepackage[russian,english]{babel}
\usepackage[noend]{algpseudocode}
\usepackage{algorithm}
\usepackage{algorithmicx}
\usepackage{color}
\usepackage[
pagebackref, 
pdfstartview=FitH,pdfpagemode=UseNone,colorlinks=true,citecolor=blue,linkcolor=blue]{hyperref}

\hypersetup{
  linkcolor=[rgb]{0.3,0.3,0.6},
  citecolor=[rgb]{0.2, 0.6, 0.2},
  urlcolor=[rgb]{0.6, 0.2, 0.2}
}
\usepackage[nameinlink,capitalise]{cleveref}

\input{thmmacros}
\input{bibmacros}
\usetikzlibrary{decorations.pathreplacing,arrows}

\newif\ifnote
\notetrue
\ifnote
\newcommand{\PHnote}[1]{\textcolor{BrickRed}{\guillemotleft PH: #1\guillemotright}}
\newcommand{\MKnote}[1]{\textcolor{Purple}{\guillemotleft MK: #1\guillemotright}}
\newcommand{\VBnote}[1]{\textcolor{NavyBlue}{\guillemotleft VB: #1\guillemotright}}
\newcommand{\CKMnote}[1]{\textcolor{magenta}{\guillemotleft CKM: #1\guillemotright}}
\else
\newcommand{\PHnote}[1]{}
\newcommand{\MKnote}[1]{}
\newcommand{\MSnote}[1]{}
\newcommand{\SBnote}[1]{}
\newcommand{\CKMnote}[1]{}
\fi

\newcommand{\var}[1]{\mathbf #1}
\newcommand{\ideal}[1]{\langle #1\rangle}
\newcommand{\der}[2]{#1^{(#2)}}

\newcommand{\hpartial}{\overline{\partial}}
\newcommand{\ehref}[1]{\href{mailto:#1}{#1}}

\newcommand{\eval}{{\mathsf{Eval}}}
\newcommand{\pts}{\mathsf{Points}}

\let\epsilon\varepsilon
\newcommand{\coeff}{\operatorname{Coeff}}

\newcommand{\vx}{\vecx}

\newcommand{\ve}{\vece}

\newcommand{\va}{\veca}

\onehalfspace

\newcommand{\rspars}{\mathsf{RowSparsity}}
\newcommand{\mme}{{multipoint evaluation }}

\title{Fast, Algebraic Multivariate Multipoint Evaluation in Small Characteristic and Applications}

\author{ 
{Vishwas Bhargava\thanks{Department of Computer Science, Rutgers University, Piscataway, NJ 08854. Research supported in part by the Simons 
Collaboration on Algorithms and Geometry and NSF grant CCF-1909683. \ehref{vishwas1384@gmail.com}. }}
\and 
{Sumanta Ghosh\thanks{ A part of this work was done during a postdoctoral stay at the Department of Computer Science \& Engineering, IIT Bombay, Mumbai, India. \ehref{besusumanta@gmail.com}.}}
\and
{Mrinal Kumar\thanks{ Department of Computer Science \& Engineering,
    IIT Bombay, Mumbai, India. \ehref{mrinal|ckm@cse.iitb.ac.in}. }}
\and
{Chandra Kanta Mohapatra\samethanks{}}
}

\date{}

\begin{document}

\maketitle

\pagenumbering{gobble}

\begin{abstract}
Multipoint evaluation is the computational task of evaluating a polynomial given as a list of coefficients at a given set of inputs. Besides being a natural and fundamental question in computer algebra on its own, fast algorithms for this problem are also closely related to fast algorithms for other natural algebraic questions like polynomial factorization and modular composition. And while \emph{nearly linear time} algorithms have been known for the univariate instance of multipoint evaluation  for close to five decades due to a work of Borodin and Moenck \cite{BM74}, fast algorithms for the multivariate version have been much harder to come by. In a significant improvement to the state of art for this problem, Umans \cite{Umans08} and Kedlaya \& Umans \cite{Kedlaya11} gave nearly linear time algorithms for this problem over field of small characteristic and over all finite fields respectively, provided that the number of variables $n$ is at most $d^{o(1)}$ where the degree of the input polynomial in every variable is less than $d$. They also stated the question of designing fast algorithms for the large variable case (i.e. $n \notin d^{o(1)}$) as an open problem. 


In this work, we show that there is a deterministic algorithm for multivariate multipoint evaluation over a field $\F_{q}$ of characteristic $p$ which evaluates an $n$-variate polynomial of degree less than $d$ in each variable on $N$ inputs in time $$\left((N + d^n)^{1 + o(1)}\poly(\log q, d, n, p)\right) \, ,$$ provided that $p$ is at most $d^{o(1)}$, and $q$ is at most $(\exp(\exp(\exp(\cdots (\exp(d)))))$, where the height of this tower of exponentials is fixed. When the number of variables is large (e.g. $n \notin d^{o(1)}$), this is the first {nearly linear} time algorithm for this problem over any (large enough) field.

Our algorithm is based on elementary algebraic ideas and this algebraic structure naturally leads to the following two independently interesting applications.
\begin{itemize}
    \item We  show that there is an \emph{algebraic} data structure for univariate polynomial evaluation with nearly linear space complexity and sublinear time complexity over finite fields of small characteristic and quasipolynomially bounded size. This  provides a counterexample to a conjecture of Milterson \cite{M95} who conjectured that over small finite fields, any algebraic data structure for polynomial evaluation using  polynomial space must have linear query complexity. 
   \item We also show that over finite fields of small characteristic and quasipolynomially bounded size,  Vandermonde matrices are not rigid enough to  yield size-depth tradeoffs for linear circuits via the current quantitative bounds in Valiant's program \cite{Valiant1977}. More precisely, for every fixed prime $p$, we show that for every constant $\epsilon > 0$, and large enough $n$, the rank of any $n \times n$ Vandermonde matrix $V$ over the field $\F_{p^a}$ can be reduced to $
   \left(n/\exp(\Omega(\poly(\epsilon)\sqrt{\log n}))\right)
   $
   by changing at most $n^{\Theta(\epsilon)}$ entries in every row of $V$, provided $a \leq \poly(\log n)$. Prior to this work, similar upper bounds on rigidity were known only for special Vandermonde matrices. For instance, the Discrete Fourier Transform matrices and Vandermonde matrices with generators in a geometric progression \cite{DL20}. 
   
\end{itemize}

\end{abstract}

\newpage

\tableofcontents 


\newpage
\pagenumbering{arabic}

\section{Introduction}

We study the question of designing fast algorithms for the following very natural and fundamental computational task. 
\begin{question}[Multipoint Evaluation]
{Given the coefficient vector of an $n$-variate polynomial $f$ of degree at most $d-1$ in each variable over a field $\F$ and  a set of points $\{\pmb\alpha_i : i \in [N]\}$ in $\F^n$, output $f(\pmb\alpha_i)$ for each $i \in [N]$.} 
\end{question}
Besides being a natural and fundamental question in computer algebra on its own, fast algorithms for this problem is also closely related to fast algorithms for other natural algebraic questions like polynomial factorization and modular composition \cite{Kedlaya11}. 

The input for this question can be specified by $(d^n + Nn)$  elements of $\F$ and clearly, there is a simple algorithm for this task which needs roughly $((d^n\cdot N)\poly(n, d))$ arithmetic operations over $\F$: just evaluate $f$ on $\pmb\alpha_i$ for every $i$ iteratively. 
Thus for $N = d^n$, the number of field operations needed by this algorithm is roughly quadratic in the input size. While \emph{nearly linear time}\footnote{Throughout this paper, we  use the phrase ``nearly linear time" to refer to algorithms such that for all sufficiently large $m$, they run in time $m^{1 + o(1)}$ on inputs of size $m$.} algorithms have been known for the univariate instance of multipoint evaluation \cite{BM74} for close to five decades, fast algorithms for the multivariate version have been much harder to come by. In a significant improvement to the state of art for this problem, Umans \cite{Umans08} and Kedlaya \& Umans \cite{Kedlaya11} gave nearly linear time algorithms for this problem over fields of small characteristic and over all finite fields respectively, provided that the number of variables $n$ is at most $d^{o(1)}$ where the degree of the input polynomial in every variable is less than $d$. They also stated the question of designing fast algorithms for the large variable case (i.e. $n \notin d^{o(1)}$) as an open problem. 

In this work, we make some concrete progress towards this question over finite fields of small characteristic (and not too large size). We also show two independently interesting applications of our algorithm. The first is  to an upper bound for algebraic data structures for  univariate polynomial evaluation over finite fields and second is to  an upper bound on the rigidity of Vandermonde matrices over fields of small characteristic. Before stating our results, we start with a brief outline of each of these problems and discuss some of the prior work and interesting open questions.  We state our results in  \autoref{sec: results}.  

\subsection{Algorithms for multivariate multipoint evaluation}

For the case of univariate polynomials and $N = d$, Borodin and Moenck \cite{BM74} showed that \mme can be solved in  $O(d\poly(\log d))$ field operations via a clever use of the Fast Fourier Transform (FFT). 


For multivariate polynomials,  when the evaluation points of interest are densely packed in a product set in $\F^n$, FFT based ideas naturally generalize to multivariate multipoint evaluation  yielding a nearly linear time algorithm. However, if the evaluation points are arbitrary and the underlying field is sufficiently large\footnote{Over small fields, for instance if $|\F| \leq d^{1 + o(1)}$ or $ |\F|^{n} \leq N^{1 + o(1)}$, a standard application of multidimensional Fast Fourier Transform which just evaluates the polynomial at all points in $\F^n$ and looks up the values at the $N$ input points works in nearly linear time. So, throughout the discussion on multipoint evaluation, we assume that $\F$ is large enough.}, and in particular not packed densely in a product set, the question of designing algorithms  for \mme that are significantly faster than the straightforward quadratic time algorithm appears to be substantially harder. In fact, the first significant progress in this direction was achieved nearly three decades after the work of Borodin and Moenck  by N\"usken and Ziegler \cite{NZ04} who showed that  for $n = 2$ and $N = d^2$, \mme can be solved in most $O(d^{\omega_2/2 + 1 })$ operations, where $\omega_2$ is the exponent for multiplying a $d\times d$  and a $d \times d^2$ matrix. The algorithm in \cite{NZ04} also generalizes to give an algorithm for general $n$ that requires $O(d^{\omega_2/2(n-1) + 1})$ field operations.\footnote{The results in both \cite{BM74} and \cite{NZ04} work for arbitrary $N$, but for simplicity have been stated for $N = d$ and $N = d^2$ respectively here. } Two significant milestones in this line of research are the results of Umans \cite{Umans08} and Kedlaya \& Umans \cite{Kedlaya11} who designed nearly linear time algorithms for this problem for fields of small characteristic and over all finite fields respectively, provided the number of variables $n$ is at most $d^{o(1)}$. We now discuss these results in a bit more detail.


Umans \cite{Umans08} gave an algorithm for \mme over finite fields of small characteristic. More precisely, the algorithm in \cite{Umans08} solves \mme in time $O((N + d^n)(n^2p)^n)\cdot \poly(d, n, p, \log N)$ over a finite field $\F$ of characteristic $p$. Thus, when $p$ and $n$ are $d^{o(1)}$, the running time can be upper bounded by $(N + d^n)^{1 + \delta}$ for every constant $\delta > 0$ and $d, N$ sufficiently large. In addition to its impressive running time, the algorithm of Umans \cite{Umans08} is also algebraic, i.e. it only requires algebraic operations over the underlying field. With \mme naturally being an algebraic computational problem, an algebraic algorithm for it has some inherent aesthetic appeal. The results in \cite{Umans08}, while being remarkable has two potential avenues for improvement, namely, a generalization to other fields and to the case when the number of variables is not $d^{o(1)}$. 

In \cite{Kedlaya11}, Kedlaya \& Umans addressed the first of these issues. They showed that \mme can be solved in nearly linear time over \emph{all} finite fields. More precisely, for every $\delta > 0$, their algorithm for \mme has running time $(N+d^n)^{1 + \delta}\log^{1 + o(1)} q$ over any finite field $\F$ of size $q$, provided $d$ is sufficiently large and $n = d^{o(1)}$. Quite surprisingly, the algorithm in \cite{Kedlaya11} is not algebraic. It goes via lifting the problem instance from the finite field $\F$ to an instance over $\Z$ and then relies on an extremely clever and unusual application of the Chinese Remainder Theorem to reduce the instance over $\Z$ back to instances over small finite fields. Intuitively, the gain in the entire process comes from the fact that in the reduced instances obtained over small finite fields, the evaluation points of interests are quite densely {packed} together inside a small product set and a standard application of the multidimensional FFT can be used to solve these small field instances quite fast. Another closely related result is a recent work of Bj\"orklund, Kaski and Williams \cite{BKW19} who (among other results) give an algorithm for multivariate multipoint evaluation but their time complexity depending polynomially on the field size (and not polynomially on the logarithm of the field size). 

In addition to these algorithms for multivariate multipoint evaluation, Umans \cite{Umans08} and Kedlaya \& Umans \cite{Kedlaya11} also show that these fast algorithms  lead to significantly faster than previously known algorithms for many other natural algebraic problems. This includes  the questions of modular composition where  the input consists of three univariate polynomials $f, g, h \in \F[X]$ of degree less than $d$ each and the goal is to output $\left(f(g(X)) \mod h(X)\right)$. In addition to being interesting on its own, faster algorithms for modular composition over finite fields are known to directly imply faster algorithms for univariate polynomial factorization over such fields. Indeed, using their nearly linear time algorithm for \mme, Umans \cite{Umans08} and Kedlaya \& Umans \cite{Kedlaya11} obtain the currently fastest known algorithms for univariate polynomial factorization over finite fields. We refer the reader to \cite{Kedlaya11} for a detailed discussion of these connections and implications. 

In spite of the significant progress on the question of algorithms for \mme in \cite{Umans08} and \cite{Kedlaya11}, some very natural related questions continue to remain open. For instance, we still do not have nearly linear time algorithms for \mme when the number of variables is large, e.g. $n \notin d^{o(1)}$ over any (large enough) finite field, or when the field is not finite. Since \mme is quite naturally an algebraic computational problem, it would also be quite interesting to have a nearly linear size arithmetic circuits over the underlying field for this problem even if such a circuit cannot be efficiently constructed. Currently, small circuits of this kind are only known over finite fields of small characteristic due to the results in \cite{Umans08}. The algorithm in \cite{Kedlaya11} does not seem to yield such a circuit since it is not algebraic over the underlying field. 


\subsection{Data structures for polynomial evaluation}\label{sec: data structures intro}

One particular implication of the results in \cite{Kedlaya11} is towards the question of constructing efficient data structures for polynomial evaluation over finite fields. The \emph{data} here is a univariate polynomial $f \in \F[X]$ of degree less than $n$ over a finite field $\F$. The goal is to process this data and store it in a way that we can support fast polynomial evaluation queries, i.e. queries of the form: given an $\alpha \in \F$ output $f(\alpha)$. The two resources of interest here are the space required to store the data and the number of locations \footnote{This can be measured  in terms of the cells accessed where each cell contains an element over the underlying field. This is an instance of the cell probe model and is quite natural in the context of algebraic data structures for algebraic problems. Alternatively, we can also measure the space and query complexity in terms of the number of bits stored and accessed respectively.} accessed for every query, i.e  the query complexity. There are two very natural solutions to this problem. 
\begin{itemize}
    \item We can store the coefficient vector of the polynomial $f$ in the memory and for each query $\alpha \in \F$, we can read the whole memory to recover the coefficient vector of $f$ and hence compute $f(\alpha)$. Thus, the space complexity and the query complexity of this data structure are both ($O(n \log q) $) bits, with clearly the space requirement being the best that we can hope for. 
    \item The second natural data structure for this problem just stores the evaluation of $f$ on all $\alpha \in \F$ in the memory, and on any query, can just read off the relevant value. Thus, the space complexity  here is $O(q \log q)$ bits, but the query complexity is $O(\log q)$ bits (which is the best that we can hope for). For $q$ being much larger than $n$ the space requirement here is significantly larger than that in the first solution. 
\end{itemize}
Using their algorithm for multipoint evaluation in \cite{Kedlaya11}, Kedlaya \& Umans construct a data structure for this problem with space complexity $n^{1 + \delta}\log^{1 + o(1)} q$ and query complexity  $\poly(\log n)\cdot \log^{1 + o(1)} q$ for all $\delta > 0$ and sufficiently large $n$. Thus, the space needed is quite close to optimal, and the query complexity is within a $\poly(\log n)$ factor of the optimal. Quite surprisingly, this data structure is not {algebraic} since it relies on the \mme algorithm in \cite{Kedlaya11} which in turn relies on non-algebraic modular arithmetic. We also note that while the algorithm for \mme over fields of small characteristic in \cite{Umans08} is algebraic, to the best of our knowledge, it does not immediately yield a data structure for polynomial evaluation. We remark that while the discussion here has been focused on data structures for univariate polynomial evaluation, the ideas in \cite{Kedlaya11} continue to work as it is even for the multivariate version of this problem and gives quantitatively similar results there. In fact, their solution to the univariate problem goes via a reduction to the multivariate case!  

In a recent work, Bj{\"o}rklund, Kaski and Williams \cite{BKW19} also prove new data structures upper bounds for polynomial evaluations for multivariate polynomials over finite fields.  These data structures are algebraic and are based on some very neat geometric ideas closely related to the notion of Kakeya sets over finite fields. Their construction can be viewed as giving a tradeoff in the space and query complexities but at least one of these parameters always appears to have polynomial dependence on the size of the underlying finite field. This is in contrast to the results in \cite{Kedlaya11} where the query complexity depends nearly linearly on $\log q$ which is more desirable for this problem. 

A very natural open question in this line of research is to obtain an algebraic data structure for this problem which matches the space and query complexity of the results in \cite{Kedlaya11}. Currently, we do not have an algebraic data structure for this problem over with even polynomial space and sublinear query complexity over any sufficiently large field. In fact, Milterson \cite{M95} showed that for algebraic data structures over finite fields of size exponential in $n$, if the space used is $\poly(n)$, then the trivial data structure obtained by storing the given polynomial as a list of coefficients and reading off everything in the memory on every query is essentially the best we can do. Milterson also conjectured a similar lower bound to hold over smaller fields. Thus, over smaller finite fields (for instance, finite fields of size $\poly(n)$), either proving a lower bound similar to that in \cite{M95} , or constructing \emph{algebraic} data structures for polynomial evaluation with perform guarantees similar to those in \cite{Kedlaya11} are extremely interesting open problems. For the later goal, it would be a good start to even have an algebraic data structure that does significantly better than the trivial solution of storing the coefficient vector of the given polynomial. 


\subsection{Non-rigidity of Vandermonde matrices}\label{sec: rigidity intro}
An application of our results for \mme is towards upper bounds for the rigidity of Vandermonde matrices. In this section, we give a brief overview of matrix rigidity.

Let $\F$ be any field. An $n \times n$ matrix $M$ over $\F$ is said to be $(r,s)$ rigid for some parameters $r, s \in \N$ if $M$ cannot be written as a sum of $n\times n$ matrices of rank at most $r$ and sparsity at most $s$. In other words, the rank of $M$ cannot be reduced to less than or equal to $r$ by changing at most $s$ of its entries. This notion was defined by Valiant \cite{Valiant1977} who showed that if the linear transformation  given by  $M$ can be computed by an arithmetic circuit of size $O(n)$ and depth $O(\log n)$, then $M$ is not $(O(n/\log\log n), O(n^{1 + \epsilon}))$ rigid for any $\epsilon > 0$. For brevity, we say that a family of matrices is \emph{Valiant rigid} if it is $(O(n/\log\log n), O(n^{1 + \epsilon}))$ rigid for some  $\epsilon > 0$. 
Thus, constructing an explicit family of matrices that are Valiant rigid suffices for proving superlinear lower bounds for log depth circuits for an explicit family of linear transformations; an extremely interesting problem that continues to be wide open. The progress on this question has been painfully slow although there have been several highly non-trivial and extremely interesting developments in this direction, e.g. \cite{ShokrollahiSS1997, Friedman1993, Lokam00, Lok01, Lokam06, GoldreichT2018, AlmanC2019, BhangaleHPT2020}. 

Even though the question of provable rigidity lower bounds for explicit matrix families has remained elusive, there has been a steady accumulation of various families of explicit matrices that are suspected to be rigid. For instance, Hadamard Matrices, Design Matrices, the Discrete Fourier Transform (DFT) matrices  and various Vandermonde Matrices have all been suspected to be rigid with varying parameters at various points in time. For some of these cases, we even have rigidity lower bounds either for special cases or with parameters weaker than what is needed for Valiant's connection to arithmetic circuit lower bounds.  However, quite surprisingly Alman \& Williams \cite{AlmanW2017} showed that Hadamard matrices are not Valiant rigid over $\Q$. This result was succeeded by a sequence of recent results all showing that many more families of matrices suspected to be highly rigid are in fact not Valiant rigid. This includes the work of Dvir \& Edelman \cite{DvirE2019}, the results of Dvir \& Liu \cite{DL20}, those of Alman \cite{Alman21}  and Kivva \cite{kivva21}. 
This list of suspected to be highly rigid that have since been proven innocent includes families like Hadamard Matrices \cite{AlmanW2017}, Discrete Fourier Transform (DFT) Matrices, Circulant and Toeplitz matrices \cite{DL20} and any family of matrices that can be expressed as a Kronecker product of small matrices \cite{Alman21, kivva21}. 

However, a notable family of matrices missing from this list is that of Vandermonde matrices. Special cases of Vandermonde matrices, for instance the DFT matrices, are known to be not be Valiant rigid, and in fact this result extends to the case of all  Vandermonde matrices where the \emph{generators} are in geometric progression.\footnote{An $n\times n$ Vandermonde matrix over a field $\F$ is specified by a list of $n$ field elements $\alpha_0, \alpha_1, \ldots, \alpha_{n-1}$ in $\F$ that we call generators. The rows and columns are indexed by $\{0, 1, \ldots, n-1\}$ and the $(i,j)$th entry of the matrix equals ${\alpha_i}^j$.} However, the case of Vandermonde matrices with arbitrary generators is still not well understood.\footnote{Lokam \cite{Lokam00} shows that $n \times n$ Vandermonde matrices with algebraically independent generators are at least $(\sqrt{n}, \Omega(n^2))$ rigid. This bound, however, is not sufficient for Valiant's program.}

\section{Our results} \label{sec: results}
We now state our results formally and try to place them in the context of prior work. 

\subsection{Algorithms for multivariate multipoint evaluation}
Our main result is a fast algebraic algorithm for \mme over fields of small characteristic. We state this result informally here, and refer the reader to \autoref{thm:polynomial-evaluation-v4} for a formal  statement. 
\begin{theorem}[Informal]\label{thm:mme informal}
Over a field $\F_{p^a}$ of characteristic $p$, there is a deterministic algorithm which evaluates a given $n$-variate polynomial of degree less than $d$ in each variable on $N$ inputs in time $$\left((N + d^n)^{1 + o(1)}\cdot \poly(a, d, n, p)\right) \, ,$$ provided that $p$ is at most $d^{o(1)}$ and  $a$ is at most $(\exp(\exp(\exp(\cdots (\exp(d)))))$, where the height of this tower of exponentials is fixed.
\end{theorem}
A few remarks are in order. 
\begin{remark} \label{rmk: description of the finite field}
    Throughout this paper, we assume that we are given a description of the field $\F_{q=p^a}$ as a part of the input. For instance, we are given an irreducible polynomial $v(Y) \in \F_p[Y]$ of degree equal to $\log_p q$ and $\F_q \equiv \F_p[Y]/\ideal{v(Y)}$. 
\end{remark}
\begin{remark}\label{rmk: algebraic algorithm}
    Our algorithms for \autoref{thm:mme informal} can be viewed as naturally giving an arithmetic circuit of nearly linear size for multivariate multipoint evaluation over the underlying finite field $\F_{p^a}$. Throughout this paper, this is what we mean when we say we have an ``algebraic" algorithm. Moreover, given a description of $\F_q$ as in \autoref{rmk: description of the finite field}, we can use the algorithm in \autoref{thm:mme informal} to  output such a circuit for multipoint evaluation in nearly linear time. 
    
\end{remark}
As alluded to in the introduction, when the number of variables is large (e.g. $n \notin d^{o(1)}$), this is the first {nearly linear} time algorithm for this problem over any sufficiently large field. Prior to this work, the fastest known algorithms for multivariate multipoint evaluation are due to the results of Umans \cite{Umans08} and Kedlaya \& Umans \cite{Kedlaya11} who give nearly linear time algorithms for this problem over finite fields of small characteristic and {all} finite fields respectively when the number of variables $n$ is at most $d^{o(1)}$. \autoref{thm:mme informal} answers an open question of Kedlaya \& Umans \cite{Kedlaya11} over the fields where it applies. For a comparison between our result, \cite{Umans08} and \cite{Kedlaya11} see Table \ref{tab:comparison}. Here, $\text{Char}(\F)$ denotes the characteristic of the field $\F$.

\begin{center}
\begin{tabular}{|p{1.5cm}|p{4.5cm}|p{2cm}|p{3.8cm}|p{2cm}|}
 \hline
 \multicolumn{5}{|c|}{Multivariate Multipoint evaluation over a finite field $\F_{q=p^a}$} \\
 \hline
 Results & Number of Operations & Algorithm type & Field constraint & Variable\\
 \hline
 \cite{Umans08}   & $(N + d^n)^{1 + \delta}$, $\forall \delta >0$, $\F_q$-operations   &Algebraic& $~~~~~~~\text{Char}(\F_q)\leq d^{o(1)}$ & $~~~~n \leq d^{o(1)}$\\\hline
\cite{Kedlaya11}&  $(N + d^n)^{1 + \delta} \log^{1+o(1)}q$, $\forall \delta >0$, many bit operations & Non-algebraic   & $~~~~~~$over any finite field & $~~~~n \leq d^{o(1)}$\\\hline
This work    &\small{$(N + d^n)^{1 + o(1)}\cdot \poly(d, n, \text{Char}(\F_q), \log q)$} $\F_q$-operations (\autoref{thm:polynomial-evaluation-v4})& Algebraic& $~~~~~~~\text{Char}(\F_q)\leq d^{o(1)}$, \small{$q \leq(\exp(\exp(\cdots(\exp (d))))$, where the height of this tower of exponentials is fixed} & $~~$ over any $n$ \\
\hline
\end{tabular}
\captionof{table}{Comparison with known results }
\label{tab:comparison}
\end{center}

By a direct connection between the complexity of multipoint evaluation and modular composition shown by Kedlaya \& Umans \cite{Kedlaya11}, \autoref{thm:mme informal} implies a nearly linear time algorithm for modular composition even when the number of variables $n$ is not less than $d^{o(1)}$. In \cite{Kedlaya11}, such an algorithm was obtained when $n < d^{o(1)}$ (over all finite fields). More precisely, we have the following corollary. 
\begin{corollary}[Informal]\label{cor:modular composition}
    Let $\F_{p^a}$ be a field of characteristic $p$. Then, there is an algorithm that on input an $n$-variate polynomial $f(X_1, X_2, \ldots, X_n)$ of individual degree less than $d$ and univariate polynomials $g_1(X), \ldots, g_n(X)$ and $h(X)$ in $\F_{p^a}[X]$ with degree  less than $N$, outputs the polynomial  $$ f(g_1(X), g_2(X), \ldots, g_n(X)) \mod h(X) $$ in time $$(d^n + N)^{1 + o(1)}\cdot \poly(a, d, n, p) \, ,$$ provided that $p$ is at most $d^{o(1)}$ and  $a$ is at most $(\exp(\exp(\exp(\cdots (\exp(d)))))$, where the height of this tower of exponentials is fixed.
\end{corollary}

Our algorithm is based on elementary algebraic ingredients. One of these ingredients is the basic fact that the restriction of a low degree multivariate polynomial to a low degree curve is a low degree univariate polynomial! We use this fact together with some other  algebraic tools, e.g. univariate polynomial interpolation (with multiplicities), structure of finite fields,  and  multidimensional FFT for our algorithm. We describe an overview of the main ideas in the proof in \autoref{sec: proof overview}. We also note that even though the algorithm in \cite{Umans08} is algebraic, it appears to be based on ideas very different from those in this paper. In particular, Umans relies on a clever reduction from the multivariate problem to the univariate problem by working over appropriate extension of the underlying field. This is then combined with the classical univariate multipoint evaluation algorithm to complete the picture. Our algorithm, on the other hand, does not involve a global reduction from the multivariate set up to the univariate set up, and crucially relies on more local properties of low degree multivariate polynomials. 

Another related prior work is a result of Bj\"orklund, Kaski and Williams \cite{BKW19}, who give a data structure (and an algorithm) for multipoint evaluation and some very interesting consequences to fast algorithms for problems in $\#\P$. We note that at a high level, the structure of our algorithm is similar to that of the algorithm of Bj\"orklund, Kaski and Williams \cite{BKW19}. However, the technical details and quantitative bounds achieved are different. One major difference is that the time complexity of the algorithm in \cite{BKW19} depends \emph{polynomially} on the field size. Thus strictly speaking, with the field size growing, this algorithm is not polynomial time in the input size. On the other hand, the time complexity of the algorithms in the works of Umans \cite{Umans08}, Kedlaya \& Umans \cite{Kedlaya11} and that in \autoref{thm:mme informal} depends polynomially in the logarithm of the field size, as is more desirable. We discuss the similarities and differences in the high level structure of the algorithm in  \cite{BKW19} and that in \autoref{thm:mme informal} in a little more detail in \autoref{sec: proof overview}. 



\subsection{Data structures for polynomial evaluation}
As an interesting application of our ideas in \autoref{thm:mme informal}, we get the following upper bound for data structure for polynomial evaluation. 
\begin{theorem}[Informal]\label{thm: data structure intro}
    Let $p$ be a fixed prime. Then, for all sufficiently large $n \in \N$ and all fields $\F_{p^a}$ with $a \leq \poly(\log n)$, there is an algebraic data structure for polynomial evaluation for univariate polynomials of degree less than $n$ over $\F_{p^a}$ that has space complexity at most $n^{1 + o(1)}$ and query complexity at most $n^{o(1)}$. 
\end{theorem}
A more precise version of \autoref{thm: data structure intro} can be found in \autoref{thm: data structure technical}. We remark that by an algebraic data structure, we mean that there is an algebraic algorithm (in the spirit of \autoref{rmk: algebraic algorithm}) over $\F_{p^a}$ that, when given the coefficients of a univariate polynomial $f$ of degree at most $n$ as input outputs the data structure ${\cal D}_f$ in time $n^{1 + o(1)}$ and another algebraic algorithm which when given an $\alpha \in \F_{p^a}$ and query access to ${\cal D}_f$ outputs $f(\alpha)$ in time $n^{o(1)}$. In other words, there is an arithmetic circuit $C_1$ over $\F_{p^a}$ with $n^{1 + o(1)}$ outputs that when given the coefficients of $f$ as input, outputs ${\cal D}_f$ and an arithmetic circuit $C_2$ with  $n^{o(1)}$ inputs satisfying the following:  for every $\alpha \in \F_{p^a}$, there is a subset $S({\alpha})$ of cells in ${\cal D}_f$ such that on input $\alpha$ and ${{\cal D}_f}|_{S(\alpha)}$, $C_2$ outputs $f(\alpha)$.

As alluded to in the introduction, Milterson \cite{M95} showed that over finite fields that are exponentially large (in the degree parameter $n$), any algebraic data structure for polynomial evaluation with space complexity $\poly(n)$ must have query complexity $\Omega(n)$. He also conjectured that the lower bound continues to hold over smaller fields.\footnote{We note that Milterson did not precisely quantify what \emph{smaller} fields mean, but the case when the field size is a large polynomial in the degree parameter $n$ is a natural setting, since the trivial data structures in this case do not have both nearly linear space and sublinear query complexity. \autoref{thm: data structure intro} provides such a construction when the underlying field additionally has a small characteristic.} \autoref{thm: data structure intro} provides a counterexample to this conjecture when the underlying field has small characteristic and is quasipolynomially bounded in size. 

The data structure of Kedlaya \& Umans \cite{Kedlaya11}  outperforms the space and query complexities of the data structure in \autoref{thm: data structure intro}. However, their construction is not algebraic; essentially because their algorithm for \mme is not algebraic.\footnote{This is also the reason why the data structure in \cite{Kedlaya11} does not give a counterexample to Milterson's conjecture.} However, their construction works over all finite fields, while we require fields of small characteristic that are quasipolynomially bounded in size. Umans' \cite{Umans08} algorithm for \mme on the other hand is algebraic, although to the best of our knowledge, this is not known to give a data structure for polynomial evaluation.  Finally, we note that for the algebraic data structure in the work of Bj{\"o}rklund, Kaski and Williams \cite{BKW19}, either the query complexity or the space complexity has polynomial dependence on the field size and thus even over fields of polynomial size it does not appear to give nearly linear space complexity or sublinear query complexity. However,  the results in \cite{BKW19} are stated for multivariate polynomials and it is not clear to us if for the special case of univariate polynomial one can somehow bypass this polynomial dependence on field size by a careful modification of their construction. 

\subsection{Upper bound on the rigidity of Vandermonde matrices}
As the second application of the ideas in \autoref{thm:mme informal}, we show the following upper bound on the rigidity of general Vandermonde matrices. 
\begin{theorem}[Informal]\label{thm: vandermonde rigidity intro}
	Let $p$ be a fixed prime. Then, for all constants $\epsilon$ with $0 < \epsilon < 0.01$ and for all  sufficiently large $n$, if $V$ is an $n \times n$ Vandermonde matrix over the field $\F_{p^a}$ for $a \leq \poly(\log n)$, then the rank of $V$ can be reduced to 	$  \frac{n}{\exp\big{(} \Omega(\epsilon^7\log^{0.5} n)\big{)}}$	by changing at most $n^{1 + \Theta(\epsilon)}$ entries of $V$. 
	
\end{theorem}
For a more formal version of \autoref{thm: vandermonde rigidity intro}, we refer to \autoref{thm: Vandermonde rigidity}. 
\autoref{thm: vandermonde rigidity intro} extends the list of natural families of matrices that were considered potential explicit candidates for rigidity but turn out to  not be  rigid enough for Valiant's program \cite{Valiant1977} of obtaining size-depth tradeoffs for linear arithmetic circuits via rigidity. Prior to this work, such upper bounds on rigidity were only known for special Vandermonde matrices, for instance, the Discrete Fourier transform matrix and Vandermonde matrices with generators in geometric progression \cite{DL20}. 

Our proof of \autoref{thm: vandermonde rigidity intro} crucially relies on the results in \cite{DL20} and combines these ideas with ideas in the proof of \autoref{thm:mme informal}. We discuss these in more details in the next section. 


\section{An overview of the proofs}\label{sec: proof overview}
In this section we describe some detail, the main high level ideas of our proofs. We begin with a detailed overview of our algorithms for multipoint evaluation. We have three algorithms (\autoref{sec: mme algo simple}, \autoref{sec: mme algo large n} and \autoref{sec: mme large fields}) starting with the simplest one and each subsequent algorithm building upon the previous one with some new ideas. We start with the simplest one here. 

\subsection{A simple algorithm for multipoint evaluation}\label{sec: overview mme algo 1}
 We start with some necessary notation. Let $p$ be a prime and $\F_q$ be a finite field with $q = p^a$. Let $f \in \F_{q}[\vecx]$ be an $n$-variate polynomial of degree at most $d-1$ in every variable and for $i = 1, 2, \ldots, N$ let $\pmb\alpha_i \in \F_q^n$ be points. The goal is to output the value of $f$ at each of these points $\pmb\alpha_i$. As is customary, we assume that the field $\F_q$ is given as $\F_p[Y]/\ideal{v(Y)}$ for some degree $a$ irreducible polynomial $v(Y) \in \F_p[Y]$.  In \autoref{lem:extracting coefficients in field representation}, we observe that given the irreducible polynomial $v(Y) \in \F_p[Y]$ such that $\F_q = \F_p[Y]/\ideal{v(Y)}$ and any $u \in \F_q$, we can efficiently compute the coefficients of the univariate polynomial over $F_p[Y]$ corresponding to $u$ via arithmetic operations over $\F_q$. Therefore, for the rest of this discussion, we assume that every field element (in the coefficients of $f$ and the coordinates of $\pmb\alpha_i$) are explicitly given to univariate polynomials of degree at most $a-1$ in $\F_p[Y]$.

We start with a discussion of the simplest version of our algorithm before elaborating on the other ideas needed for further improvements. The formal guarantees for this version can be found in \autoref{thm:polynomial-evaluation-v1}. The algorithm can be thought to have two phases, the preprocessing phase and the local computation phase. 

\paragraph*{Preprocessing phase. }We start with a description of the preprocessing phase. 
\begin{itemize}
    \item \textbf{A subfield of appropriate size: } As the first step of the algorithm, we compute a natural number $b$ such that $p^{b-1} \leq adn \leq p^b$. For the ease of this discussion, let us assume that $b$ divides $a$, and thus $\F_{p^b}$ is a subfield of $\F_q = \F_{p^a}$. If $b$ does not divide $a$, then we work in a field $\F_{p^c}$ that is a common extension of $\F_{p^a}$  and $\F_{p^b}$. 
    \item \textbf{Evaluating $f$ on $\F_{p^b}^n$: } We now use the standard multidimensional Fast Fourier Transform algorithm to evaluate $f$ on all of $\F_{p^b}^n$. This algorithm runs in quasilinear time in the input size, i.e. $\tilde{O}(d^n + (p^{bn}))$, where $\tilde{O}$ hides $\poly(d, n, p, b)$ factors. From our choice of $b$, we note that this quantity is at most $\tilde{O}((padn)^n)$. 
    
\end{itemize}

\paragraph*{Local computation phase. }We now describe the local computation phase. 
\begin{itemize}

    \item \textbf{A low degree curve through $\pmb\alpha_i$: }   Once we have the evaluation of $f$ on all points in $\F_{p^b}^n$, we initiate some \emph{local} computation at each $\pmb\alpha_i$. This local computation would run in time $(adn)^c$ for some fixed constant $c$, thereby giving an upper bound of $\tilde{O}\left((pad)^n + N(adn)^{O(1)} \right)$ on the total running time. To describe this local computation, let us focus on a point $\pmb\alpha_i$. Since the field elements of $\F_q$ are represented as univariate polynomials of degree at most $(a-1)$ in $\F_p[Y]$, we get that for every $\pmb\alpha_i \in \F_q^n$, there exist vectors $\pmb\alpha_{i,0}, \pmb\alpha_{i,1}, \ldots, \pmb\alpha_{i,a-1}$ in $\F_p^n$ such that 
    \[
    \pmb\alpha_i = \pmb\alpha_{i,0} + \pmb\alpha_{i,1}Y + \cdots + \pmb\alpha_{i,a-1}Y^{a-1} \, .
    \]
    Let us now consider the curve $\vecg(t) \in \F_p^n[t]$ defined as 
    \[
    \vecg_i(t) = \pmb\alpha_{i,0} + \pmb\alpha_{i,1}t + \cdots + \pmb\alpha_{i,a-1}t^{a-1} \, .
    \]
    We are interested in some simple properties of this curve. The first such property is that it passes through the point $\pmb\alpha_i$, since $\pmb\alpha_i = \vecg_i(Y)$ (recall that $Y$ is an element of $\F_q = \F_p[Y]/\ideal{v(Y)}$ here). The second property is that this curve contains \emph{a lot} of points in the $\F_{p^b}^n$. In particular, note that for every $\gamma \in \F_{p^b}$, $\vecg_i(\gamma) \in \F_{p^b}^n$. Thus, there are at least $p^b$ points on $\vecg_i(t)$ in $\F_{p^b}^n$ (counted with multiplicities). 
    
    \item \textbf{Restriction of $f$ to $\vecg_i(t)$: } We now look at the univariate polynomial $h_i(t)$ obtained by restricting the $n$-variate polynomial $f$ to the curve $\vecg_i(t)$. Thus, if $\vecg_i(t) = (g_{i,0}(t),\ldots, g_{i,n-1}(t))$ for some univariate polynomials $g_{i,j}(t)$ of degree at most $a-1$, then $h_i(t)$ is equal to the polynomial $f(g_{i,0}(t),\ldots, g_{i,n-1}(t))$. Clearly, the degree of $h_i$ is at most $a(d-1)n < adn$. From our previous discussion, we know that $h_i(Y) = f(\pmb\alpha_i)$. Moreover, we have already evaluated $f$ on all of $\F_{p^b}^n$ and thus, we know the value of $h_i(\gamma)$ for all $\gamma \in \F_{p^b}$. Note that these are at least $p^b$ many inputs on which the value of $h_i(t)$ is correctly known to us. Also, from our choice of $b$, we know that $p^b > adn > \deg(h_i)$. Thus, we can recover the polynomial $h_i$ completely using univariate polynomial interpolation in time at most $\poly(a, d, n, p)$, and thus can output $h_i(Y) = f(\pmb\alpha_i)$ in time $\poly(a,d,n,p)$. Iterating this local computation for every $i \in \{0,1,\ldots, N-1\}$, we can compute the value of $f$ at $\pmb\alpha_i$ for each such $i$. 
    \end{itemize}
\paragraph*{Correctness and running time. } The correctness of the algorithm immediately follows from the outline above. Essentially, we set things up in a way that to compute $f(\pmb\alpha_i)$ it suffices to evaluate the univariate polynomial $h_i$ at input $Y \in \F_q$. Moreover, from the preprocessing phase, we already have the value of $f$ on $\F_{p^b}^n$ and this in turn gives us the evaluation of $h_i(t)$ on $p^b > adn > \deg(h_i)$ distinct inputs. Thus, by standard univariate polynomial interpolation, we recover $h_i$ and hence $h_i(Y) = f(\pmb\alpha_i)$ correctly. 

The time complexity of the preprocessing phase is dominated by the step where we evaluate $f$ on $\F_{p^b}^n$. This can be upper bounded by $\tilde{O}((padn)^n)$ using the standard multidimensional FFT algorithm. In the local computation phase, the computation at each input point $\pmb\alpha_i$ involves constructing the curve $\vecg_i(t)$, constructing the set $\{(\gamma, h_i(\gamma)): \gamma \in \F_{p^b}\}$, using the evaluation of $h_i$ on these $p^b$ inputs to recover $h_i$ uniquely via interpolation and then computing $h_i(Y)$. For every $\gamma \in \F_{p^b}$, $\vecg_i(\gamma) \in \F_{p^b}^n$ can be done in time at most $\poly(a, d, n, p)$. So, the total time complexity of this phase is at most $(N\cdot  \poly(a, d, n, p))$, and hence the total running time of the algorithm is $\tilde{O}(N + (padn)^n)$.




\subsection{Towards faster multipoint evaluation}
The algorithm outlined in the previous section achieves a  ${O}(Nn + d^n)^{1 + o(1)}$ when $apn = d^{o(1)}$. We now try to modify it so that it continues to be nearly linear time even when the number of variables $n$ and the degree of underlying field $a$ are not less than $d^{o(1)}$. The factor of $p^n$ appears to be inherent to our approach and seems difficult to get rid of, and this leads to the restriction of working over fields of small characteristic for all our results in this paper.


Before proceeding further, we remark that the basic intuition underlying all of our subsequent algorithms are essentially the same as those in the simple algorithm outlined in this section. For each of the further improvements, we modify certain aspects of this algorithm using a few more technical (and yet simple) ideas on top of the ones already discussed in \autoref{sec: overview mme algo 1}.

\paragraph*{Handling large number of variables. }
The factor of $n^n$ in the running time appears in the preprocessing phase of the algorithm in \autoref{sec: overview mme algo 1}. The necessity for this stems from the fact that the univariate polynomial $h_i(t)$ obtained by restricting $f$ to the curve $\vecg_i(t)$ through $\pmb\alpha_i$ can have degree as large as $a(d-1)n$. Thus, for interpolating $h_i(t)$ from its evaluations, we need its value on at least $a(d-1)n + 1$ distinct inputs. Thus, we need $p^b$ to be at least $a(d-1)n + 1$.

However, we note that if we have access to not just the evaluations of $h_i(t)$, but also to the evaluations of its derivatives up to order $n-1$ at each of these inputs in $\F_{p^b}$, then $h_i(t)$ can be uniquely from this information provided $p^b$ is at least $\deg(h_i(t))/n$, i.e. $(a(d-1)n + 1)/n \leq ad$ (see \autoref{lem:hermitian-interpolation} for a formal statement). Thus, with observation at hand, we now choose $b$ such that $p^{b-1} \leq ad \leq p^b$. Moreover, for the local computation, we  now need not only the evaluation of $h_i$ on all points in $\F_{p^b}$ but also the evaluations of all derivatives of $h_i(t)$ of order at most $n-1$ on all these points. A natural way of ensuring that the evaluations of these derivatives of $h_i(t)$ are available in the local computation phase is to compute not just the evaluation  of $f$ but also of \emph{all} its partial derivatives of up to $n$ on all of $\F_{p^b}^n$. Together with the chain rule of partial derivatives, we can use the evaluations of these partial derivatives of $f$ and the identity $h_i(t) = f \circ \vecg_i(t)$ to obtain the evaluations of $h_i(t)$ and all its derivatives of order at most $n-1$ on all inputs in $\F_{p^b}$. This ensures that $h_i$ can once again be correctly and uniquely recovered given this information via a standard instance of Hermite Interpolation, which in turn ensures the correctness of the algorithm.     

To see the effect on the running time, note that in the preprocessing phase, we now need to evaluate not just $f$ but all its partial derivatives of order at most $n-1$ on all of $\F_{p^b}^n$. Thus, there are now roughly $\binom{n + n}{n} \leq 4^n$ polynomials to work with in this phase. So, given the coefficients of $f$, we first obtain the coefficients of all these derivatives, and then evaluate these polynomials on $\F_{p^b}^n$ using a multidimensional FFT algorithm again. Also, the coefficient representation of any fixed derivative of order up to $n-1$ can be computed from the coefficients of $f$ in $\tilde{O}(d^n)$ time (see \autoref{lem:computing-hasse-derivation}). Thus, the total time complexity of the preprocessing phase in this new algorithm can be upper bounded by $\tilde{O}((adp)^n4^n)$. 

Once we have this stronger guarantee from the preprocessing phase, we get to doing some local computation at each point $\pmb\alpha_i$. Now, instead of recovering $h_i$ via a standard univariate polynomial interpolation, we have to rely on a standard Hermite interpolation for this. In particular, we need access to the evaluation of all derivatives of $h_i(t)$ of order at most $n-1$ on all inputs  $\gamma \in \F_{p^b}$. This can be done via an application of chain rule of derivatives and the fact that we have evaluations of \emph{all} partial derivatives of $f$ of order at most $n-1$ on all points in $\F_{p^b}^n$. The time taken for this computation at each $\gamma \in \F_{p^b}$ turns out to be about ${O}(4^n\poly(d, n, a, p))$. Thus, the total time taken for local computation at all the input points can be upper bounded by  roughly ${O}(N4^n\poly(d, n, a, p))$. 

Thus, the total time complexity of this modified algorithm is $\tilde{O}((N + (adp)^n)4^n)$. In other words, we have managed to remove the factor of $n^n$ present in the algorithm in \autoref{sec: overview mme algo 1} and replace it by $4^n$. An algorithm based on this improvement is described in \autoref{sec: mme algo large n}.

\paragraph*{Handling larger fields. }
We now discuss the improvement in the dependence on the parameter $a$, which is the degree of the extension of $\F_p$ where the input points lie. In the local computation step at each point, the curve $\vecg_i(t)$ through $\pmb\alpha_i$ has degree $a-1$ in the worst case, since we view the field elements in $\F_{p^a}$ as univariate polynomials of degree at most $a-1$ with coefficients in $\F_p$. Therefore, the restriction of $f$ to such a curve, namely the polynomial $h_i(t)$ can have degree $(a-1)\deg(f)$ in the worst case. This forces us to choose the parameter $b$ such that $p^b$ is at least $\deg(h_i)$, thereby leading to a factor of $a^n$ in the running time. Note that if we had the additional promise that the point $\pmb\alpha_i$ was in an extension $\F_{p^{\tilde{a}}}$ of $\F_p$ for some $\tilde{a} < a$, then the curve $\vecg_i$ would be of degree at most $(\tilde{a}-1) < (a-1)$ and hence the polynomial $h_i$ would have degree at most $(\tilde{a}-1)\deg(f)$. More generally, if all the input points $\pmb\alpha_i$ were promised to be in $\F_{p^{\tilde{a}}}^n$, we can improve the factor $a^n$ to $(\tilde{a})^n$ in the running time by choosing $b$ such that $p^b$ is larger than $\tilde{a}dn$ (in fact, we only need $p^b \geq (\tilde{a}d)$ if we are working with multiplicities). We also note that for every $\tilde{a} \in \N$ the curve $\vecg_i(t)$ takes a value in $\F_{p^{\tilde{a}}}^n$ whenever $t$ is set to a value in $\F_{p^{\tilde{a}}}$. As a consequence, the curve $\vecg_i$ contains at least $p^{\tilde{a}}$ points in $\F_{p^{\tilde{a}}}^n$. With these observations in hand, we now  elaborate  on the idea for reducing the $a^n$ factor in the running time. For simplicity of exposition, we outline our ideas in the setting of the algorithm discussed in \autoref{sec: overview mme algo 1}. In particular, derivative based improvements are not involved. 

Let $a'$ be such that $p^{a'} > adn \geq p^{a'-1}$. Now, instead of recovering $h_i$ directly from its values  on $\F_{p^b}$, we try to recover $h_i$ in two steps. In the first step, we try to obtain the values of $h_i(\gamma)$ for every $\gamma \in \F_{p^{a'}}$ using  the information we have from the preprocessing phase. Assuming that we can do this, we can again obtain $h_i$ by interpolation and compute $h_i(Y) = f(\pmb\alpha_i)$. 

Now, to compute $h_i(\gamma)$ for $\gamma \in \F_{p^{a'}}$, we note that $h_i(\gamma)$ equals $f\circ \vecg_i(\gamma)$, thus it would be sufficient if we had the evaluation of $f$ on the point set $\{\vecg_i(\gamma) : \gamma \in \F_{p^{a'}}\}$. This seems like the problem we had started with, but with one key difference: the points $\{\vecg_i(\gamma) : \gamma \in \F_{p^{a'}}\}$ are all in $\F_{p^{a'}}^n$ with $a' = \Theta(\log adn)$! Thus, the degree of the extension where these points lie is significantly reduced. In essence, this discussion gives us a reduction from the problem of evaluating $f$ on $N$  points in $\F_{p^a}^n$ to evaluating $f$ on $N\cdot adn$ points in $\F_{p^{a'}}^n$, with $a' = \Theta(\log adn)$. Thus, we have another instance of multipoint evaluation with a multiplicatively larger point set in an extension of $\F_p$ of degree logarithmic in $adn$. If we now apply the algorithm discussed in \autoref{sec: overview mme algo 1}, we get a running time of roughly $\tilde{O}(Nadn + (pdn\log (adn))^n)$. Thus, in the running time, the factor $a^n$ has been replaced by $\log^n a$ at the cost of $N$ being replaced by $Nadn$. In fact, we can continue this process $\ell$ times, and in each step we end up with an instance of multipoint evaluation with the size of the point set being increased by a multiplicative factor, with the gain being that we have a substantial reduction in the degree of the field extension that the points live in. 

This idea can be combined with those used for improving the dependence on the number of variables, to get our final algorithm that achieves nearly linear running time provided that $p = d^{o(1)}$ and $a\leq \exp(\exp(\ldots(\exp(d))))$ where the height of this tower of exponentials is fixed. We refer to  \autoref{thm:polynomial-evaluation-v4} for a formal statement of the result and \autoref{sec: mme large fields} for further details. 


\paragraph*{Comparison with the techniques of Bj\"orklund, Kaski and Williams \cite{BKW19}. }\label {sec: bkw comparison} Now that we have an overview of the algorithms for multipoint evaluation in this paper, we can elaborate on the similarities they share with the algorithms in \cite{BKW19}. At a high level, the similarities are significant. In particular, both the algorithms have a preprocessing phase where the polynomial is on a product set using multidimensional FFT. This is followed by a local computation step, where the value of the polynomial at any specific input of interest is deduced from the already computed data by working with the restriction of the  multivariate polynomial to an appropriate curve. In spite of these similarities in the high level outline, the quantitative details of these algorithms are different. One salient difference is that the time complexity of the algorithm in \cite{BKW19}, depends polynomially on the size of the underlying field, whereas in our algorithm outlined above, this dependence is polynomial in logarithm of the field size as long as the size of the field is bounded by a tower function of fixed height in the degree parameter $d$. This difference stems from technical differences in the precise product set used in the preprocessing phase and the sets of curves utilized in the local computation phase. In particular, the degree of the curves in the local computation phase of our algorithms depends polynomially on $\log |\F|$, where as the degree of the curves used in \cite{BKW19} depends polynomially on $|\F|$. Additionally, algorithms in  \cite{BKW19} rely on the assumption that the total degree of the polynomial divides $|\F^{*}| -1$, whereas we do not need any such divisibility condition. 

\subsection{Data structure for polynomial evaluation}\label{sec: data structures overview}
The multipoint evaluation algorithm in \autoref{thm:mme informal}  is  naturally conducive to obtaining data structures for polynomial evaluation. Essentially, the evaluation of the polynomial in a fixed grid (independent of the $N$ points of interest in the input) gives us the data structure, and the local computation at each input point of interest which requires access to some of the information computed in the preprocessing phase constitutes the query phase of the data structure. We discuss this in some more detail now. 

Let $f(X) \in \F_{p^a}[X]$ be a univariate polynomial of degree at most $n$. We start by picking parameters $d, m$ such that $d^m$ is at least $n$. For any such choice of $d$ and $n$, there is clearly an $m$-variate polynomial $F(Z_0, Z_1, \ldots, Z_{m-1})$ such that $F(X, X^d, X^{d^2}, \ldots, X^{d^{m-1}}) = f(X)$. In other words, the image of $F$ under the Kronecker substitution equals $f$. Now, as in the multipoint evaluation algorithms, we pick the smallest integer  $b$ such that $p^b > adm$ and evaluate $F$ on $\F_{p^b}^m$ and store these points along with the value of $F$ on these inputs in the memory. This forms the memory content of our data structure. Thus, the memory can be thought of having $p^{bm} \leq (padm)^m$ cells, each containing a pair $(\vecc, F(\vecc))$ for $\vecc \in \F_{p^b}^m$. 

Let us now consider the query complexity of this data structure. Let $\alpha \in \F_{p^a}$ be an input and the goal is to compute $f(\alpha)$. From the relation between $F$ and $f$, we have that $f(\alpha) = F(\pmb\alpha)$, where $\pmb\alpha = (\alpha, \alpha^d, \alpha^{d^2}, \ldots, \alpha^{d^{m-1}})$. Now, we rely on the local computation in the multipoint evaluation algorithms to compute $F(\pmb\alpha)$. In the algorithm, we consider a curve $\vecg$ of degree at most $a-1$ which passes through $\pmb\alpha$ and look at the restriction of $F$ to this curve to get a univariate polynomial $h$ of degree less than $adm$. Then, we take the value of $h$ on inputs in $\F_{p^b}$, which can be recovered from the value of $F$ on the points in the set $\vecg(\F_{p^b}) \cap \F_{p^b}^n$. Finally, note that there are at least $p^b > adm$ of these inputs and value of $h$ on these inputs is  already stored in the memory. This suffices to recover $h$ and thus, also $f(\alpha) = F(\pmb\alpha)$. So, the query complexity of this data structure is $adm$. 

To get a sense of the parameters, let us set $d = n^{1/\log\log n}$ and $m = \log\log n$. Clearly, the constraint $d^m \geq n$ is met in this case. For this choice of parameter and for $p$ being a constant and $a \leq \poly(\log n)$, we get that the space complexity is at most $n^{1 + o(1)}$ and the query complexity is at most $n^{o(1)}$. 

The complete details can be found in \autoref{sec:data structures full}.

\subsection{Rigidity of Vandermonde matrices}\label{sec: rigidity proof overview}
The connection between rigidity of Vandermonde matrices and multipoint evaluation is also quite natural. Consider a Vandermonde matrix $V_n$ with generators $\alpha_0, \ldots, \alpha_{n-1}$ and for every $i, j \in \{0,1, \ldots, n-1\}$, the $(i,j)$th entry of $V_n$ is $\alpha_i^{j}$. Now, for any univariate polynomial $f$ of degree at most $n-1$, the coefficients of $f$, together with the set $\{\alpha_i : i \in \{0,1, \ldots, n-1\}\}$ of generators form an instance of (univariate) multipoint evaluation. Moreover, for any choice of the generators $\{\alpha_i : i \in \{0,1, \ldots, n-1\}\}$, the algorithm for multipoint evaluation, e.g \autoref{thm:mme informal} can naturally be interpreted as a circuit for computing the linear transform given by the matrix $V_n$. Furthermore, if this linear circuit is structured enough, we could, in principle hope to get a decomposition of $V_n$ as a sum of a sparse and a low rank matrix from this linear circuit, for instance, along the lines of the combinatorial argument of Valiant \cite{Valiant1977}. Our proof of \autoref{thm: vandermonde rigidity intro} is along this outline. We now describe these ideas in a bit more detail. 

Given a univariate polynomial $f$ of degree $n-1$ and inputs $\alpha_0, \alpha_1, \ldots, \alpha_{n-1}$, let $F$ be an $m$- variate polynomial of degree $d$ such that $(n = d^{m})$\footnote{For simplicity, let us assume that such a choice of integers $d, m$ exist.} as described in \autoref{sec: data structures overview}. Moreover, for $i \in \{0,1, \ldots, n-1\}$, let $\pmb\alpha_i = (\alpha_i, \alpha_i^d, \ldots, \alpha_i^{d^{m-1}})$. Now, as discussed in \autoref{sec: data structures overview}, $f(\alpha_i) = F(\pmb\alpha_i)$. Let $\Tilde{V}$ be the $n \times n$ matrix where the rows are indexed by $\{0,1, \ldots, n-1\}$ and the columns are indexed by all $m$- variate monomials of individual degree at most $d-1$. We use the fact that $d^m = n$ here. From the above set up, it immediately follows that the coefficient vectors of $f$ and $F$ are equal to each other (with the coordinate indices having slightly different semantics) and the matrices $V_n$ and $\Tilde{V}$ are equal to each other. 

We now observe that the algorithm for multipoint evaluation described in \autoref{sec: overview mme algo 1} gives a natural decomposition of $\Tilde{V}$ (and hence $V_n$) as a product of a matrix $A$ of row sparsity at most $adm$ and a $p^{bm} \times d^m$ matrix $B$ with $b$ being the smallest integer such that $p^b > adm$. The rows of $B$ are indexed by all elements of $\F_{p^b}^m$ and the columns are indexed by all $m$-variate monomials of individual degree at most $d-1$, and the $(\pmb\alpha, \ve)$ entry of $B$ equals $\pmb\alpha^{\ve}$. Intuitively, the matrix $B$ corresponds to the preprocessing phase of the algorithm and the matrix $A$ corresponds to the local computation. At this point, we use  an upper bound of \cite{DL20} on the rigidity of Discrete Fourier Transform matrices over finite fields and the inherent Kronecker product structure of the matrix $B$ to obtain an upper bound on the rigidity of $B$. Finally, we observe that that matrix $V_n = \Tilde{V} = A\cdot B$ obtained by multiplying a sufficiently non-rigid matrix $B$ with a row sparse matrix $A$ continues to be non-rigid with an interesting regime of parameters. This essentially completes the proof. For more details, we refer the reader to \autoref{sec:rigidity proof}. \vspace{-0.5cm}
\section{Preliminaries}
We use $\N$ to denote the set of natural numbers $\{0,1,2,\ldots\}$, $\F$ to denote a general field. For any positive integer $N$, $[N]$ denotes the set $\{1,2,\ldots, N\}$. By $\var x$ and $\var z$, we denote the variable tuples $(X_1,\ldots,X_n)$ and $(Z_1,\ldots, Z_n)$, respectively. For any $\var e=(e_1,\ldots, e_n)\in \N^n$, $\var x^\var e$ denotes the monomial $\prod_{i=1}^nX_i^{e_i}$. By $|\var e|_1$, we denote the sum $e_1+\cdots+e_n$. 

For every positive integer $k$, $k!$ denotes $\prod_{i=1}^ki$. For $k=0$, $k!$ is defined as $1$. For two non-negative integer $i$ and $k$ with $k\geq i$, $\binom{k}{i}$ denotes $\frac{k!}{i!(k-i)!}$. For $k<i$, $\binom{k}{i}=0$. For non-negative integer $i_1,\ldots, i_s$ with $i_1+\cdots+i_s=k$, $\binom{k}{i_1,\ldots,i_s}=\frac{k!}{i_1!\cdots i_s!}$. For $\var a=(a_1,\ldots,a_n), \var b=(b_1,\ldots,b_n)\in\N^n$, $\binom{\var a}{\var b}=\prod_{i=1}^n\binom{a_i}{b_i}$, and $\binom{\var a+\var b}{\var a,\var b}=\prod_{i=1}^n\binom{a_i+b_i}{a_i,b_i}$.

We say that a function $\psi:\N \to \N$ is polynomially bounded, or denoted by $\psi(n) \leq \poly(n)$, if there exists a constant $c$ such that for all large enough $n \in \N$, $\psi(n) \leq n^c$.
\begin{proposition}
\label{prop:binomial-estimation}
For any two positive integers $i$ and $k$ with $k\geq i$, $$\binom{k}{i}\leq \left(\frac{ke}{i}\right)^i.$$
\end{proposition}
For proof see \cite[Chapter 1]{Jukna}. Suppose that $p$ be a positive integer greater than $1$. Then for any non-negative integer $c$, $\log_p^{\circ c}(n)$ denotes the $c$-times composition of logarithm function with itself, with respect to base $p$. For example, $\log^{\circ 2}_p(n)=\log_p\log_p(n)$. By $\log^{\star}_p(n)$, denotes the smallest non-negative integer $c$ such that $\log^{\circ c}_p(n)\leq 1$. For $p=2$, we may omit the subscript $p$ in $\log_p(n)$, $\log^{\circ c}_p(n)$ and $\log^{\star}_p(n)$.

\subsection{Some facts about finite fields}
Suppose that $p$ is a prime and $q=p^a$ for some positive integer $a$. Then there exists an \emph{unique} finite field of size $q$. In other words, all the finite fields of size $q$ are \emph{isomorphic} to each other. We use $\F_q$ to denote the finite field of size $q$, and $p$ is called the characteristic of $\F_q$. For any finite field $\F_q$, $\F^*_q$ represents the multiplicative cyclic group after discarding the field element $0$. For any irreducible polynomial $v(Y)$ over $\F_q$, the quotient ring $\F_q[Y]/\ideal{v(Y)}$ forms a larger field over $\F_q$ of size $q^b$ where $b$ is the degree of $v(Y)$. The next lemma describes that we can efficiently construct such larger fields over $\F_q$, when the characteristic of the field is small.
\begin{lemma}
\label{lem:finite-field-construction}
Let $p$ be a prime and $q=p^a$ for some positive integer $a$. Then, for any positive integer $b$, the field $\F_{q^b}$ can be constructed as $\F_q[Y]/\ideal{v(Y)}$, where $v(Y)$ is degree $b$ irreducible polynomial over $\F_q$, in $\poly(a,b,p)$ $\F_q$-operations. Furthermore, all the basic operations in $\F_{q^b}$ can be done in $\poly(b)$ $\F_q$-operations.
\end{lemma}
\begin{proof}
The elements of the quotient ring $\F_q[Y]/\ideal{v(Y)}$ are polynomials in $Y$ over $\F_q$ with degree less than $b$, and the operations are polynomial addition and multiplication under modulo $v(Y)$. Therefore, once we have an irreducible $v(Y)$ (over $\F_q$) of degree $b$, we can perform the basic operations in $\F_{q^b}$ using $\poly(b)$ $\F_{q}$-operations. From \cite[Theorem 4.1]{Shoup90}, we can compute a degree $b$ irreducible polynomial $v(Y)$ over $\F_q$ using $\poly(a,b,p)$ $\F_p$-operations. 
\end{proof}

Fix a field $F_q$ of characteristic $p$. In the standard algebraic model over $\F_q$, the basic operations are addition, subtraction, multiplication, and division of elements in $\F_q$. Let  $\F_q = \F_p[X]/\ideal{g(X)}$  where $q=p^a$ and $g(X)$ is a degree $a$ irreducible polynomial over $\F_p$. Then for any element $\alpha \in \F_q$, consider its canonical representation $\alpha=\alpha_0+\alpha_1 X + \ldots + \alpha_{a-1}X^{a-1}$ where $\alpha_i \in \F_p$. Note that it is not clear how to extract $\alpha_i$'s from $\alpha$ using the algebraic operations over $\F_q$.  We show that this is possible if $p$ is small. Since $\F_q =\F_p[X]/\ideal{g(X)}$, $X \in \F_q$ is a root of the degree $a$ irreducible polynomial $g(X)$ (over $\F_p$). This implies that  $X, X^p, X^{p^2}, \ldots, X^{p^{a-1}}$ are all distinct elements of $\F_q$. 

\begin{lemma}\label{lem:extracting coefficients in field representation}
Let $p$ be prime and $q=p^a$ for some positive integer $a$. Let $\F_q=\F_p[X]/\ideal{g(X)}$ where $g(X)$ is a degree $a$ irreducible polynomial over $\F_p$. Let $\alpha\in\F_q$ and $\alpha=\alpha_0+\alpha_1X+\cdots+\alpha_{a-1}X^{a-1}$ where $\alpha_i\in\F_p$. Then, given blackbox access to $\alpha$ and $\F_q$-operations,  $\alpha_0,\alpha_1,\ldots, \alpha_{a-1}$ can be computed in $ \poly(a,\log{p})$ $\F_q$-operations.
\end{lemma}

\begin{proof}
 Note that, given $\alpha$, we can compute $\alpha^p$ by repeated squaring over $\F_q$. Applying this iteratively, we have access to all conjugates $\alpha, \alpha^p, \alpha^{p^2}, \ldots, \alpha^{p^{a-1}}$. Observe that, 
    $$ \underbrace{\begin{bmatrix}
1 & X & X^2 & \dots & X^{a-1} \\
1 & X^p & X^{2p} & \dots & X^{p(a-1)} \\
\dots  & \dots  & \dots  & \dots & \dots  \\
1 & X^{p^{a-1}} & X^{2p^{a-1}} & \dots & X^{(a-1) p^{a-1}}  \end{bmatrix}}_A
\begin{bmatrix}
\alpha_0 \\ \alpha_1 \\ \dots \\ \alpha_{a-1} 
\end{bmatrix}
=
\begin{bmatrix}
\alpha \\ \alpha^p \\ \dots \\ \alpha^{p^{a-1}}
\end{bmatrix}.$$

Note that, the matrix $A$ in the above linear system is a Vandermonde matrix and thus invertible. Also, each entry of $A$ is an element in $\F_q$. Thus, we can find $\alpha_i$ by solving this linear system over $\F_q$. For time complexity, note that we can use $\alpha^{p^i}$ to compute $\alpha^{p^{i+1}}$. Thus, $\alpha , \alpha^p, \ldots , \alpha^{p^{a-1}}$ can be computed in  $\poly(a, \log p )$ $\F_q$-operations. Also, the computation of $A$ and solving the linear system can be done in $\poly(a, \log p )$ $\F_q$-operations. Therefore, overall complexity is $\poly(a, \log p)$ $\F_q$-operations. 
\end{proof}

Thus, for the rest of our paper, we  consider that the extraction of the $\F_p$-coefficients from elements in $\F_q$ as an algebraic operation. Also, in our applications, the time complexity overhead introduced due to this is negligible.

Suppose that $\F_{q_1}$ and $\F_{q_2}$ are two finite fields of characteristic $p$ such that $\F_{q_1}$ is a subfield of $\F_{q_2}$. Then $\F_{q_2}$ forms a vector space over $\F_{q_1}$. A subset $\{\beta_1,\beta_2,\ldots,\beta_k\}$ of $\F_{q_2}$ is called an $\F_{q_1}$-basis if every element of $\alpha\in F_{q_2}$ is a unique linear combination of $\beta_i$'s over $\F_{q_1}$. 
\begin{lemma}
\label{lem:subfield-construction}
Let $p$ be a prime and $q=p^a$ for some positive integer $a$. Let $b$ be a positive integer and $\F_{q^b}=\F_q[Y]/\ideal{v(Y)}$ for some degree $b$ irreducible polynomial $v(Y)$ over $\F_q$. Then, the following holds:
\begin{enumerate}
    \item The field $\F_{q^b}$ contains the subfield $\F_{p^b}$. Furthermore, all the elements of $\F_{p^b}$ can be computed in $p^b\cdot \poly(a,b,p)$ $\F_q$-operations. 
    \item In $\poly(a, b, p)$ $\F_q$-operations, an element $\beta\in\F_{q^b}$ can be computed such that $\{1,\beta,\ldots,\beta^{b-1}\}$ forms an $\F_p$-basis for $\F_{p^b}$. Moreover, given any element $\alpha\in \F_{p^b}$, the $\F_p$-linear combination of $\alpha$ in the basis  $\{1,\beta,\ldots,\beta^{b-1}\}$ can be computed in $\poly(b)$ $\F_q$-operations.
\end{enumerate}
\end{lemma}
\begin{proof}
 Since $p^b-1$ divides $q^b-1$, $\F_{q^b}$ is a splitting field of the $x^{p^b}-x$, that is, $x^{p^b}-x$ linearly factorizes over $\F_{q^b}$. Now, one can show that the roots of $x^{p^b}-x$ over $\F_{q^b}$  form a subfield of size $p^b$. Now, using \cite[Theorem 3.2]{Shoup90}, we can compute a degree $b$ irreducible polynomial $u(Z)$ over $\F_p$ in $\poly(b,p)$ $\F_p$-operations. Next, applying \cite{Berlekamp70}, we can find a root $\beta\in \F_{q^b}$ for $u(Z)$ in $\poly(a,b,p)$ $\F_q$-operations. One can show that for any other polynomial $u'(Z)$ with $u'(\beta)=0$, $u(Z)$ divides $u'(Z)$. Also, $\beta$ is in $\F_{p^b}$ since $\F_{p^b}$ is a splitting field for $u(Z)$. This implies that $\{1,\beta,\ldots,\beta^{b-1}\}$  forms an $\F_p$-basis for $\F_{p^b}$. Thus, after having $\beta$, we can compute all the elements of $\F_{p^b}$ by taking all possible $\F_p$-linear combinations of $\{1,\beta,\ldots,\beta^{b-1}\}$. The cost of doing this is $p^b\cdot\poly(b,p)$ $\F_q$-operations. Computing $\beta$ takes $\poly(a,b,p)$ $\F_q$-operations. Therefore, in $p^b\cdot \poly(a,b,p)$ $\F_q$-operations, we can compute all the elements of $\F_{p^b}$.
 
 Let $\alpha\in \F_{p^b}$. Since $\F_{p^b}$ is a subfield of $\F_{q^b}$, from the representation of $\F_{q^b}$, we can write $\alpha=\alpha_0+\alpha_1Y+\cdots+\alpha_{b-1}Y^{b-1}$ where $\alpha_i\in\F_q$. Also, for all $i\in\{0,1,\ldots, b-1\}$, each $\beta^i$ can be written as $\beta_{i,0}+\beta_{i,1}Y+\cdots+\beta_{i, b-1}Y^{b-1}$ where $\beta_{i,j}\in\F_q$. Let $\alpha=c_0+c_1\beta+\cdots+c_{b-1}\beta^{b-1}$, where $c_i$'s are unknown and we want to find them. This combined with the representation of $\alpha$ and $\beta^i$, we get a system of linear equations in $\{c_0,\ldots, c_{b-1}\}$ over $\F_q$. Now we can solve it in $\poly(b)$ $\F_q$-operations and get $c_i$'s.
\end{proof}






\subsection{Hasse derivatives}
In this section, we briefly discuss the notion of Hasse derivatives that plays a crucial role in our results. 
\begin{definition}[Hasse derivative]
\label{def:hasse-derivative}
Let $f(\var x)$ be an $n$-variate polynomial over a  field $\F$. Let $\var e=(e_1,\ldots, e_n)\in\N^n$. Then,  the Hasse derivative of $f$ with respect to the monomial $\var x^{\var e}$ is the coefficient of $\vecz^{\vece}$ in the polynomial $f(\vecx + \vecz) \in (\F[\vecx])[\vecz]$. 
\end{definition}
\paragraph*{Notations.} Suppose that $f(\var x)$ be an $n$-variate polynomial over a field $\F$. Let $\var b\in \N^n$. Then, $\hpartial_{\var b}(f)$ denotes the Hasse derivative of $f(\var x)$ with respect to the monomial $\var x^{\var b}$. For any non-negative integer $k$,  $\hpartial^{\leq k}(f)$ is defined as 
\[
\hpartial^{\leq k}(f) = \left\{ \hpartial_{\var b}(f)\,\mid\,\var b\in\N^n \text{ s.t. }|\var b|_1\leq k \right\} ,
\]
and $ \hpartial^{< k}(f)$ denotes the set $\{\hpartial_{\var b}(f)\,\mid\,\var b\in\N^n \text{ s.t. } |\var b|_1< k\}$. 

For a univariate polynomial $h(t)$ over $\F$ and a non-negative integer $k$, $\der{h}{k}(t)$ denotes the Hasse derivative of $h(t)$ with respect to the monomial $t^k$, that is, $\coeff_{Z^k}(h(t+Z))$.

Next, we mention some useful properties of Hasse derivatives.  
\begin{proposition}
\label{lem:hasse-derivative-property}
Let $f(\var x)$ be an $n$-variate polynomial over $\F$. Let $\var a,\var b\in\N^n$. Then, 
\begin{enumerate}
    \item $\hpartial_{\var a}(f)=\sum_{\var e\in\N^n}\binom{\var e}{\var a}\coeff_{\var x^{\var e}}(f)\var x^{\var e-\var a}$.
    \item $\hpartial_{\var a}\hpartial_{\var b}(f)=\binom{\var a+\var b}{\var a,\var b}\hpartial_{\var a+\var b}(f).$
\end{enumerate}
\end{proposition}
For proof one can see \cite[Appendix C]{F14}. The following lemma describes the cost of computing Hasse derivatives.
\begin{lemma}
\label{lem:computing-hasse-derivation}
Let $p$ be a prime and $q=p^a$ for some positive integer $a$. Let $f(\var x)$ be an $n$-variate polynomial over $\F_q$ with individual degree less than $d$. Let $\var b=(b_1,\ldots,b_n)\in\N^n$. Then, given $f(\var x)$ and $\var b$ as input, Algorithm \ref{algo:computing-hasse-derivative} outputs $\hpartial_{\var b}(f)$ in $$d^n\cdot\poly(n)+\poly(b,d)$$ $\F_q$-operations, where $b=\max_{i\in[n]}\, b_i$.
\end{lemma}
\begin{proof}
We first describe the algorithm and then argue about its correctness and running time. 
\begin{algorithm}[H]
	\caption{Computing Hasse derivative}
	\label{algo:computing-hasse-derivative}
	\textbf{Input:} An $n$-variate polynomial $f(\var x)\in\F_q[\var x]$ with individual degree less than $d$ and $\var b=(b_1,\ldots,b_n)\in\N^n$.\\
	\textbf{Output:} $\hpartial_{\var b}(f)$.\\
	\begin{algorithmic}[1]
        \State Let $b$ be $\max_{i\in [n]}\, b_i$.
          \State Let $D$ be an $(b+1)\times d$ array.
          \For{$j\leftarrow 0$ to $d-1$}
              \For{$i\leftarrow 0$ to $b$}
                 \If{$i=j$}
                    \State $D_{i,j}\leftarrow 1$.
                  \ElsIf{$i>j$}
                    \State $D_{i,j}\leftarrow 0$.
                  \ElsIf{$i=0$}
                    \State $D_{0,j}\leftarrow 1$.
                  \Else
                    \State $D_{i,j}=D_{i-1, j-1}+D_{i,j-1}$
                   \EndIf
              \EndFor
          \EndFor
          \For{$\var e\in\{0,1,\ldots,d-1\}^n$}
             \State Let $\var e=(e_1,\ldots, e_n)$.
             \State $c_{\var e}\leftarrow \coeff_{\var x^{\var e}}(f)\cdot\prod_{i=1}^n D_{b_i,e_i}$. 
          \EndFor
          \State Output $\sum_{\var e\in\{0,1,\ldots, d-1\}^n}c_{\var e}\var x^{\var e-\var b}$.
    \end{algorithmic}
\end{algorithm}
In Algorithm \autoref{algo:computing-hasse-derivative}, for all $i\in\{0,1,\ldots,b\}$ and $j\in\{0,1,\ldots,d-1\}$, the $(i,j)$th entry of array $D$ $$D_{i,j}=\binom{j}{i}\, \mathrm{mod}\, p.$$ For this, we note that the arithmetic in Line 15 of the algorithm is happening over the underlying field $\F_q$. This combined with \autoref{lem:hasse-derivative-property} implies that the Algorithm \autoref{algo:computing-hasse-derivative} computes $\hpartial_{\var b}(f)$. 

To compute the array $D$, we are performing $d(b+1)$ $\F_p$-operations. Computing all $c_{\var e}$'s for $\var e\in\{0,1,\ldots,d-1\}^n$ takes $d^n\cdot (n+1)$ $\F_q$-operations. Therefore, Algorithm \autoref{algo:computing-hasse-derivative} runs in our desired time complexity.
\end{proof}

\subsection{Univariate polynomial evaluation and interpolation}
The two simplest but most important ways of representing an univariate polynomial of degree less than $d$ are either by giving the list of its coefficients, or by giving its evaluations at $d$ distinct points. In this section, we discuss about the cost of changing between these two representations. First, we mention the cost of polynomial evaluation, that is, going from the list of coefficients to the list of evaluations.
\begin{lemma}[Evaluation]
\label{lem:univariate-evaluation}
Let $f(x)$ be a degree $d$ polynomial over $\F$. Let $\alpha_1,\alpha_2,\ldots,\alpha_N$ be $N$ distinct elements from $\F$. Then, $f(\alpha_i)$ for all $i\in[N]$ can be computed in $O(Nd)$ $\F$-operations.
\end{lemma}
For each $i\in[N]$, using Horner's rule, one can compute $f(\alpha_i)$ with $d-1$ additions and $d-1$ multiplications over $\F$. Therefore, the total cost of computing $f(\alpha_i)$ for all $i\in[N]$ is $O(Nd)$ operations. For more details see \cite[Section 5.2]{GG03}. Next, we discuss the cost of polynomial interpolation where we go from the list of evaluations to the list of coefficients.  
\begin{lemma}[Interpolation]
\label{lem:univariate-interpolation}
Let $f(x)$ be a degree $d$ polynomial over $\F$. Let $\alpha_0,\alpha_1,\ldots,\alpha_d$ be $(d+1)$ distinct elements from $\F$. Let $\beta_i=f(\alpha_i)$ for all $i\in\{0,1,\ldots, d\}$. Then, given $(\alpha_i,\beta_i)$ for all $i\in\{0,1,\ldots, d\}$, $f(x)$ can be computed in $O(d^2)$ $\F$-operations.
\end{lemma}
For proof see \cite[Section 5.2]{GG03}. The following lemma gives a stronger version of univariate polynomial interpolation, known as Hermite interpolation. Here, the number of evaluation points can be less than $d$, but  evaluations of Hasse derivatives of the polynomial up to certain order is available.    
\begin{lemma}[Hermite interpolation]
\label{lem:hermitian-interpolation}
Let $f(x)$ be a degree $d$ univariate polynomial over a field $\F$ and $e_1,\ldots, e_m$ be $m$ positive integers such that $e_1+\cdots+e_m$ is greater than $d$. Let $\alpha_1,\ldots, \alpha_m$ be $m$ distinct elements from $\F$. For all $i\in[m]$ and $j\in[e_j]$, let $\der{f}{j-1}(\alpha_i)=\beta_{ij}$. Then given $(\alpha_i,j,\beta_{ij})$ for all $i\in[m]$ and $j\in[e_j]$, $f(x)$ can be computed in $O(d^2)$ $\F$-operations.
\end{lemma}
For proof see \cite[Section 5.6]{GG03}. We also remark that while there are nearly linear time algorithms for all of the above operations (multipoint evaluation, interpolation and Hermite interpolation) based on the Fast Fourier transform; however, for our applications in this paper, the above stated  more naive bounds suffice. 

\subsection{Multidimensional Fast Fourier transform}
We  crucially rely on the following lemma that says that there is a fast  algorithm for evaluating an $n$-variate polynomial $f$ with coefficients in a finite field $\F$ on the  set $\tilde{\F}^{n}$ where $\tilde{\F}$ is a subfield of $\F$. The proof is based on a simple induction on the number of variables and uses the standard FFT one variable at a time. For the proof, see Theorem $4.1$ in \cite{Kedlaya11}. 
\begin{lemma}\label{lem: multidim FFT}
Let $\F$ be a finite field and let $\tilde{\F}$ be a subfield of $\F$. Then, there is a deterministic algorithm that takes as input an $n$-variate polynomial $f \in \F[\vx]$ of degree at most $d-1$ in each variable as a list of coefficients, and in at most $(d^n + |\tilde{\F}|^n)\cdot \poly(n, d, \log |\F|)$ operations over the field $\F$, it outputs the evaluation of $f$ for all  $\pmb\alpha \in \tilde{\F}^n$.
\end{lemma}


\section{A simple algorithm for multipoint evaluation}\label{sec: mme algo simple}
We start with our first and simplest algorithm for \mme. The algorithm gives an inferior time complexity to what is claimed in \autoref{thm:mme informal}, but contains some of the main ideas. Subsequently, in \autoref{sec: mme algo large n} and \autoref{sec: mme large fields}, we build upon this algorithm to eventually prove \autoref{thm:mme informal}. Our main theorem for this section is the following. 
\begin{theorem}
\label{thm:polynomial-evaluation-v1}
Let $p$ be a prime and $q=p^a$ for some positive integer $a$. There is a deterministic algorithm such that on input an $n$-variate polynomial $f(\var x)$ over $\F_q$ with individual degree less than $d$ and points  $\pmb\alpha_1,\pmb\alpha_2,\ldots, \pmb\alpha_N$ from $\F_q^n$, it outputs  $f(\pmb\alpha_i)$ for all $i\in[N]$ in time 
$$(N+(adnp)^n)\cdot\poly(a,d,n,p) \, .$$

\end{theorem}
\subsection{A description of the algorithm}
We start with a description of the algorithm, followed by its analysis. We recall again that through all the algorithms in this and subsequent sections, we assume that the underlying field $\F_q$ is given to us via an irreducible polynomial of appropriate degree over the prime subfield. Moreover, from \autoref{lem:extracting coefficients in field representation}, we also assume without loss of generality that for every input field element, we have access to its representation as a polynomial of appropriate degree over the prime subfield. For a polynomial map $\vecg(t) = (g_1(t), g_2(t), \ldots, g_n(t))$ and an $n$-variate polynomial $f$, we use $f(\vecg(t))$ to denote the univariate polynomial $f(g_1(t), g_2(t), \ldots, g_n(t))$. 
	\begin{algorithm}[H]
	\caption{Efficient Multivariate Multipoint  Evaluation}
	\label{algo:polynomial-evaluation-v1}
	\textbf{Input:} An $n$-variate polynomial $f(\var x)\in\F_q[\var x]$ with individual degree less than $d$ and $N$ distinct points $\pmb\alpha_1,\pmb\alpha_2,\ldots,\pmb\alpha_N$ from $\F_q^n$.\\
	\textbf{Output:} $f(\pmb\alpha_1),f(\pmb\alpha_2),\ldots,f(\pmb\alpha_N)$.\\
	\begin{algorithmic}[1]
		\State Let $p$ be the characteristic of $\F_q$ and $q=p^a$.
		\State Let $v_0(Y_0)$ be an irreducible polynomial in $\F_p[Y_0]$ of degree $a$ and  $$\F_{q}=\F_p[Y_0]/\ideal{v_0(Y_0)}.$$ 
		\State Let $b$ be the smallest integer such that $p^b>adn$.
		\State Compute an irreducible polynomial $v_1(Y_1)$ in $\F_q[Y_1]$ of degree $b$ and  $$\F_{q^b}=\F_q[Y_1]/\ideal{v_1(Y_1)}.\text{ (\autoref{lem:finite-field-construction})}$$ 
		\State Compute the subfield $\F_{p^b}$ of  $\F_{q^b}$. (\autoref{lem:subfield-construction})
		\State Evaluate $f(\var x)$ over the grid $\F_{p^b}^n$. (\autoref{lem: multidim FFT})
		\For{all $i\in[N]$}
		    \State Let $\pmb\alpha_i=\pmb\alpha_{i,0}+\pmb\alpha_{i,1}Y_0+\cdots+\pmb\alpha_{i,a-1}Y_0^{a-1}$, where $\pmb\alpha_{i,j}\in \F_{p}^n$. 
		    \State Let $\var g_{i}(t)$ be the curve defined as $\pmb\alpha_{i,0}+\pmb\alpha_{i,1}t+\cdots+\pmb\alpha_{i,a-1}t^{a-1}.$
		    \State Compute the set $P_i=\{(\gamma,\var g_i(\gamma)) \,\mid\, \gamma\in\F_{p^b}\}$. (\autoref{lem:univariate-evaluation})
		    \State Compute the set $E_i=\{(\gamma, f(\pmb\gamma'))\,\mid\, (\gamma,\pmb\gamma')\in P_i\}$ from the evaluations of $f(\var x)$ over $\F_{p^b}^n$.
		    \State Let $h_i(t)$ be the univariate polynomial  defined as $f(\var g_{i}(t))$. 
		    \State Using $E_{i}$, interpolate $h_{i}(t)$. (\autoref{lem:univariate-interpolation})
		    \State Output $h_{i}(Y_0)$ as $f(\pmb\alpha_i)$. (\autoref{lem:univariate-evaluation}) 
		\EndFor
	\end{algorithmic}
\end{algorithm}	

\subsection{Analysis of Algorithm \autoref{algo:polynomial-evaluation-v1}}
\begin{proof}[Proof of \autoref{thm:polynomial-evaluation-v1}]
We start with the proof of correctness of the algorithm. 
\paragraph*{Correctness of Algorithm \autoref{algo:polynomial-evaluation-v1}. }
We show that Algorithm \autoref{algo:polynomial-evaluation-v1} computes $f(\pmb\alpha_i)$ for all $i\in[N]$ in $(N+(adnp)^n) \cdot \poly(a,d,n,p)$ many $\F_q$ operations. We assume that the underlying field $\F_q$ is represented as $\F_{p}[Y_0]/\ideal{v_0(Y_0)}$, where $v_0(Y_0)$ is a degree $a$ irreducible polynomial over $\F_p$. From \autoref{lem:finite-field-construction}, the field $\F_{q^b}$ can be constructed as $\F_q[Y_1]/\ideal{v_1(Y_1)}$ for some degree $b$ irreducible polynomial $v_1(Y_1)$ over $\F_q$.  \autoref{lem:subfield-construction} ensures that we can explicitly compute all the elements of the subfield $\F_{p^b}$ (of $\F_{q^b}$). The representation of $\F_q$ ensures that every element $\beta\in\F_q$ is of the form $\beta_0+\beta_1Y_0+\cdots+\beta_{a-1}Y_0^{a-1}$, where $\beta_i\in\F_p$. Therefore, for all $i\in[N]$, $\pmb\alpha_i$ is of the form $\pmb\alpha_{i,0}+\pmb\alpha_{i,1}Y_0+\cdots+\pmb\alpha_{i,a-1}Y_0^{a-1}$, where $\pmb\alpha_{i,j}\in\F_p^n$. For all $i\in[N]$, the curve $\var g_i(t)$ is defined as $\pmb\alpha_{i,0}+\pmb\alpha_{i,1}t+\cdots+\pmb\alpha_{i,a-1}t^{a-1}$. Since $f$ is an $n$-variate polynomial over $\F_q$ with individual degree less than $d$, for all $i\in[N]$, the polynomial $h_i(t)=f(\var g_i(t))$ is a polynomial in $t$ of degree less than $adn$. For all $\gamma\in\F_{p^b}$, $\var g_i(\gamma)$ is in $\F_{p^b}^n$. Therefore, from the evaluations of $f(\var x)$ over the grid $\F_{p^b}^n$, we get the set $E_i=\{(\gamma,h_i(\gamma)\,\mid\, \gamma\in\F_{p^b}\}$. Since the degree of $h_i$ is less than $adn$ and $p^b$ is greater than $adn$, from the set $E_i$, we can interpolate $h_i(t)$. The construction of $\var g_i(t)$ ensures that $\var g_i(Y_0)=\pmb\alpha_i$. Hence, $h_i(Y_0)=f(\pmb\alpha_i)$ for all $i\in[N]$. 

\paragraph*{Time complexity of Algorithm \autoref{algo:polynomial-evaluation-v1}. }Now we discuss the time complexity of Algorithm \autoref{algo:polynomial-evaluation-v1}. From \autoref{lem:finite-field-construction}, the field $\F_{q^b}$ can be constructed as $\F_q[Y_1]/\ideal{v_1(Y_1)}$ for some degree $b$ irreducible polynomial $v_1(Y_1)$ over $\F_q$ in $\poly(a,b,p)$ many $\F_{q}$-operations. Also, all the basic operations in the field $\F_{q^b}=\F_q[Y_1]/\ideal{v_1(Y_1)}$ can be done using $\poly(b)$ $\F_q$-operations. Applying \autoref{lem:subfield-construction}, the cost of computing all the elements of the subfield $\F_{p^b}$ (of $\F_{q^b}$) is $p^b\cdot \poly(a,b, p)$ $\F_q$-operations. Using \autoref{lem: multidim FFT}, we can evaluate $f(\var x)$ over the grid $\F_{p^b}^n$ in $$(d^n+p^{bn})\cdot\poly(a,b,d, n, p)$$ $\F_q$-operations. For all $i\in[N]$, using \autoref{lem:univariate-evaluation}, the cost of computing the set $P_i=\{(\gamma,\var g_i(\gamma))\,\mid\, \gamma\in\F_{p^b}\}$ is $p^b\cdot\poly(a,b,n)$  $\F_{q}$-operations. Using the set $E_i$, \autoref{lem:univariate-interpolation} ensures that $h_i(t)$ can be interpolated using $\poly(a,b,d,n)$ $\F_q$-operations. Finally, $h(Y_0)$ can be computed in $\poly(a,d,n)$ many $\F_q$-operations. Since $adn<p^b\leq adnp$, the above discussion implies that that Algorithm \ref{algo:polynomial-evaluation-v1} performs $$(N+(adnp)^n)\cdot \poly(a,d,n,p)$$ $\F_q$-operations.
\end{proof}

\section{Multipoint evaluation for large number of variables} \label{sec: mme algo large n}
In this section, we append the overall structure of Algorithm \autoref{algo:polynomial-evaluation-v1} with some more ideas to improve the dependence of the running time on $n$. In particular, the goal is to reduce the $n^n$ factor in the running time of \autoref{thm:polynomial-evaluation-v1} to a factor of the form $\exp(O(n))$. The main result of this section is the following theorem. 
\begin{theorem}
\label{thm:polynomial-evaluation-v2}
Let $p$ be a prime and $q=p^a$ for some positive integer $a$. There is a deterministic algorithm such that on input an $n$-variate polynomial $f(\var x)$ over $\F_q$ with individual degree less than $d$ and points  $\pmb\alpha_1,\pmb\alpha_2,\ldots, \pmb\alpha_N$ from $\F_q^n$, it outputs  $f(\pmb\alpha_i)$ for all $i\in[N]$ in time 
$$(N+(adp)^n)\cdot 4^n \cdot \poly(a,d,n,p). \,$$


\end{theorem}
A useful additional ingredient in the proof of this theorem is the following lemma. Semantically, this is an explicit form of the chain rule of Hasse derivatives for the restriction of a multivariate polynomial to a curve of low degree.

\begin{lemma}
\label{lem:hasse-derivative}
Let $f(\var x)$ be an $n$-variate degree $d$ polynomial over a field $\F$, $\var g(t)=(g_1,\ldots, g_n)$ where $g_i\in\F[t]$, and $h(t)=f(\var g(t))$. For all $i\in[n]$, let $g_i(t+Z)=g_i(t)+Z\tilde g_i(t, Z)$ for some $\tilde{g_i} \in \F[t, Z]$. Let $\tilde{\var g}(t,Z)=(\tilde g_1,\ldots, \tilde g_n)$, and for all $\var e=(e_1,\ldots,e_n)\in\N^n$, $\tilde{\var g}_{\var e}=\prod_{i=1}^n\tilde g_i^{e_i}$. For any $\ell\in \N$, let $$h_{\ell}(t,Z)=\sum_{i=0}^\ell Z^i\sum_{\var e\in\N^n:|\var e|_1=i}\hpartial_{\var e}(f)(\var g(t))\cdot\tilde{\var g}_{\var e}(t,Z).$$ Then, for every $k\in\N$ with $k\leq \ell$, $\der{h}{k}(t)=\coeff_{Z^k}(h_\ell)$.
\end{lemma}
\begin{proof}
By the definition of Hasse derivative, $\der{h}{k}(t)= \coeff_{Z^k}(h(t+Z))$. On the other hand, 
\begin{align*}
h(t+Z) &= f(g_1(t+Z),\ldots, g_n(t+Z))\\
&=f(g_1+Z\tilde g_1,\ldots, g_n+Z\tilde g_n).
\end{align*}
Applying Taylor's expansion on $f(g_1+Z\tilde g_1,\ldots, g_n+Z\tilde g_n)$, we get that 
\begin{align*}
h(t+Z) &= \sum_{i=0}^dZ^i\sum_{\var e\in\N^n:|\var e|_1=i}\hpartial_{\var e}(f)(\var g(t))\cdot \tilde{\var g}_{\var e}\\
&=h_{\ell}+ \sum_{i=\ell+1}^dZ^{i}\sum_{\var e\in\N^n: |\var e|_1=i}\hpartial_{\var e}(f)(\var g(t))\cdot\tilde{\var g}_{\var e}.
\end{align*}
The lowest possible degree of $Z$ in the second part of the above sum is greater than $\ell$. Therefore, the coefficient of $Z^k$ in the second part is zero since $k\leq \ell$. Hence, $$\der{h}{k}(t)= \coeff_{Z^k}(h(t+Z))=\coeff_{Z^k}(h_{\ell}(t,Z)),$$ which completes the proof. 
\end{proof}
\subsection{A description of the algorithm}
We start by describing the algorithm, followed by its analysis. 
\begin{algorithm}[H]
	\caption{Efficient multivariate polynomial evaluation with large number of variables}
	\label{algo:polynomial-evaluation-v2}
	\textbf{Input:} An $n$-variate polynomial $f(\var x)\in\F_q[\var x]$ with individual degree less than $d$ and $N$ points $\pmb\alpha_1,\pmb\alpha_2,\ldots,\pmb\alpha_N$ from $\F_q^n$.\\
	\textbf{Output:} $f(\pmb\alpha_1),f(\pmb\alpha_2),\ldots,f(\pmb\alpha_N)$.\\
	\begin{algorithmic}[1]
		\State Let $p$ be the characteristic of $\F_q$ and $q=p^a$.
		\State Let $v_0(Y_0)$ be an irreducible polynomial in $\F_p[Y_0]$ of degree $a$ and   $$\F_q=\F_p[Y_0]/\ideal{v_0(Y_0)}.$$ 
		\State Let $b$ the smallest positive integer such that $p^b>ad$.
		\State Compute an irreducible polynomial $v_1(Y_1)$ in $\F_q[Y_1]$ of degree $b$ and  $$\F_{q^b}=\F_q[Y_1]/\ideal{v_1(Y_1)}. \text{ (\autoref{lem:finite-field-construction})}$$ 
		\State Compute the subfield $\F_{p^b}$ of  $\F_{q^b}$. (\autoref{lem:subfield-construction})
		\State Compute the set $\hpartial^{<n}(f)$. (\autoref{lem:computing-hasse-derivation})
		\State Evaluate all the polynomials in $\hpartial^{<n}(f)$ over the grid $\F_{p^b}^n$. (\autoref{lem: multidim FFT})
		\For{all $i\in[N]$}\label{step:algo-v2-forloop}
		    \State Let $\pmb\alpha_i=\pmb\alpha_{i,0}+\pmb\alpha_{i,1}Y_0+\cdots+\pmb\alpha_{i,a-1}Y_0^{a-1}$, where $\pmb\alpha_{i,j}\in \F_{p}^n$. 
		    \State Let $\var g_{i}(t)$ be the curve defined as $\pmb\alpha_{i,0}+\pmb\alpha_{i,1}t+\cdots+\pmb\alpha_{i,a-1}t^{a-1}$.
		    \State Let $h_i(t)=f(\var g_i(t))$.
		    \State Let $E_i=\{(\gamma,\der{h_i}{0}(\gamma), \der{h_i}{1}(\gamma),\ldots, \der{h_i}{n-1}(\gamma)\,\mid\, \gamma\in\F_{p^b}\}$.
		    \State Invoke the function \textsc{Evaluate Derivatives A } with input $\var g_i(t)$ and compute the set $E_i$.
		    \State Using $E_i$, interpolate $h_i(t)$. (\autoref{lem:hermitian-interpolation})\label{label:v2-interpolate} 
		    \State Output $h_i(Y_0)$ as $f(\pmb\alpha_i)$. (\autoref{lem:univariate-evaluation}) 
		\EndFor
		\Statex
	\end{algorithmic}
\end{algorithm}
We now describe the function Evaluate Derivatives A invoked above. We follow the same notation as in Algorithm \autoref{algo:polynomial-evaluation-v2} including the local variable names.
\begin{algorithm}[H]
	\caption{Function to generate data for Hermite Interpolation}
	\label{algo:polynomial-evaluation-v2-eval-derivatives}
	\begin{algorithmic}[1]
		\Function{Evaluate Derivatives A }{$\var g(t)$}
		\State Let $\var g(t)=(g_1,\ldots, g_n)$.
		\State For all $i\in[n]$, let $g_i(t+Z)=g_i(t)+Z\tilde g_i(t,Z)$ and $\tilde{\var g}(t,Z)=(\tilde g_1(t,Z),\ldots, \tilde g_n(t,Z))$.
		\State Compute $\tilde g_i(t,Z)$ for all $i\in[n]$. (\autoref{lem:computing-hasse-derivation})
	    \State For all $\var e=(e_1,\ldots, e_n)\in\N^n$, let $\tilde{\var g}_{\var e}=\prod_{i=1}^n\tilde g_i^{e_i}$.
	    \State Compute the set of polynomials $\{\tilde{\var g}_{\var e}(t,Z)\,\mid\,|\var e|_1<n\}$. (Polynomial multiplication)
	    \State $P\leftarrow \emptyset$.
		\For{all $\gamma \in\F_{p^b}$}\label{label:algo-v2-function-forloop}
		    \State Using evaluations of polynomials in $\hpartial^{<n}(f)$ over $\F_{p^b}^n$, compute the polynomial $$h_\gamma(Z)=\sum_{i=0}^{n-1}Z^i\sum_{\var e\in\N^n:|e|_1=i}\hpartial_{\var e}(f)(\var g(\gamma))\tilde{\var g}_{\var e}(\gamma, Z).$$\label{label:algo-v2-polynomial}
		    \State For all $i\in\{0,1,\ldots, n-1\}$, extract $\coeff_{Z^{i}}(h_{\gamma})$.
		    \State $P\leftarrow P\cup\{(\gamma, \coeff_{Z^0}(h_\gamma),\coeff_{Z^1}(h_\gamma),\ldots, \coeff_{Z^{n-1}}(h_\gamma))\}$.
		 \EndFor
		 \State \Return $P$.
		\EndFunction
	\end{algorithmic}
\end{algorithm}

\subsection{Analysis of Algorithm \autoref{algo:polynomial-evaluation-v2}}
\begin{proof}[Proof of \autoref{thm:polynomial-evaluation-v2}]
We start with the proof of correctness. 
\paragraph*{Correctness of Algorithm \autoref{algo:polynomial-evaluation-v2}. }
We show that Algorithm \autoref{algo:polynomial-evaluation-v2} computes $f(\pmb\alpha_i)$ for all $i\in[N]$ in the desired time.  Like Algorithm \autoref{algo:polynomial-evaluation-v1}, we assume that the underlying field $\F_{q}$ is represented as $\F_{p}[Y_0]/\ideal{v_0(Y_0)}$  for some degree $a$ irreducible polynomial $v_0(Y_0)$ over $\F_p$. However, unlike Algorithm \autoref{algo:polynomial-evaluation-v2}, here we pick $b$ as the smallest positive integer satisfying $p^b>ad$. Like Algorithm \autoref{algo:polynomial-evaluation-v1}, here also we construct a degree $b$ extension $\F_{q^b}$ over $\F_q$ and compute all the elements of the subfield $\F_{p^b}$ (of $\F_{q^b}$). \autoref{lem:finite-field-construction} ensures that we can construct $\F_{q^b}$ as $\F_q[Y_1]/\ideal{v_1(Y_1)}$ for some degree $b$ irreducible polynomial over $\F_q$ and from \autoref{lem:subfield-construction}, we can compute $\F_{p^b}$. The crucial difference with Algorithm \autoref{algo:polynomial-evaluation-v1} is the way we interpolate the polynomial $h_i(t)$ in Line $14$ of Algorithm \autoref{algo:polynomial-evaluation-v2}. 
The field $\F_{p^b}$ may have much smaller number of points than the degree of $h_i(t)$. Therefore, to interpolate $h_i(t)$, we have to evaluate all the Hasse derivatives of $h_i(t)$ up to order $n-1$ at points in $\F_{p^b}$. Next we describe the correctness of Algorithm \autoref{algo:polynomial-evaluation-v2} in detail.

As mentioned in the proof of \autoref{thm:polynomial-evaluation-v1}, the representation of $\F_{q}$ ensures that each $\pmb\alpha_i$ is of form $\pmb\alpha_{i,0}+\pmb\alpha_{i,1}Y_0+\cdots+\pmb\alpha_{i, a-1}Y_0^{a-1}$ where $\pmb\alpha_{i,j}$ is in $\F_p^n$. Therefore, $h_i(t)=f(\var g_i(t))$ is a polynomial of degree less than $adn$ since $f$ is an $n$-variate polynomial with individual degree less than $d$. This implies that, like Algorithm \ref{algo:polynomial-evaluation-v1}, we can interpolate the $h_i(t)$ by evaluating it at $adn$ many distinct points. However, for the choice of $b$ in this algorithm, we don't have $adn$ elements.  So, for all $\gamma\in\F_{p^b}$, we compute the evaluations at $\gamma$ of all Hasse derivatives of $h_i$ up to order $n-1$ and invoke \autoref{lem:hermitian-interpolation} with this data to recover $h_i$. From \autoref{lem:hasse-derivative}, using the evaluations of $\hpartial^{<n}(f)$ over the grid $\F_{p^b}^n$, the function \textsc{Evaluate Derivatives A} of Algorithm \ref{algo:polynomial-evaluation-v2-eval-derivatives} computes the set $E_i$ consisting of $n${th} order derivative information of $h_i$ for every $\gamma \in \F_{p^b}$. Given this set $E_i$, we invoke \autoref{lem:hermitian-interpolation} to successfully interpolate $h_i(t)$ and output $f(\pmb\alpha_i)=h_i(Y_0)$ for all $i\in[N]$.

\paragraph*{Time complexity of Algorithm \autoref{algo:polynomial-evaluation-v2}.}
We now  describe the time complexity of Algorithm \ref{algo:polynomial-evaluation-v2}. Similar to Algorithm \ref{algo:polynomial-evaluation-v1}, using \autoref{lem:finite-field-construction}, the construction of $\F_{q^b}$ takes $\poly(a,b,p)$ $\F_q$-operations and its basic operations can be done in $\poly(b)$ $\F_{q}$-operations. From \autoref{lem:subfield-construction}, all the elements of $\F_{p^b}$ can be computed in $p^b\cdot \poly(a,b,p)$ $\F_q$-operations. Since $\binom{n+n-1}{n}$ is upper bounded by $4^n$, applying \autoref{lem:computing-hasse-derivation}, the set of polynomials $\hpartial^{<n}(f)$ can be computed in $(4d)^n\cdot \poly(d,n)$ $\F_q$-operations. Hence, from \autoref{lem: multidim FFT}, computing all the polynomials in $\hpartial^{<n}(f)$ over the grid $\F_{p^b}^n$ requires $$(d^n+p^{bn})\cdot 4^n\cdot\poly(a,b,d,n,p)$$ $\F_q$-operations.  Next we discuss the time taken by the \emph{for} loop in Algorithm \ref{algo:polynomial-evaluation-v2} at Line $8$.


First we estimate the cost of each iteration of the loop.  For that, we need to analyze the complexity of the function \textsc{Evaluate Derivatives A}. The input $\var g(t)=(g_1,\ldots, g_n)$ to \textsc{Evaluate Derivatives A} is a curve of degree at most $a-1$. Using \autoref{lem:computing-hasse-derivation}, $\tilde{\var g}(t,Z)=(\tilde g_1,\ldots, \tilde g_n)$ can be computed in $\poly(a,n)$ $\F_q$-operations. Thus, the total cost of computing the set $\{ \tilde{\var g}_{\var e}(t,Z)\,\mid\, |\var e|_1<n\}$ is $4^n\cdot \poly(a,n)$ $\F_q$-operations. Given $\gamma\in\F_{p^b}$, we can evaluate $\tilde{\var g}_{\var e}(t,Z)$ at $t=\gamma$ in $\poly(a,b,n)$ $\F_q$-operations. Thus, for each $\gamma\in\F_{p^b}$, the polynomial $h_{\gamma}(Z)$ at Line $9$ in Algorithm \ref{algo:polynomial-evaluation-v2-eval-derivatives} can be computed at the cost of $4^n\cdot \poly(a,b,n)$ $\F_q$-operations. After computing $h_{\gamma}(Z)$ as its list of coefficients, we collect the coefficients of $Z^i$ of $h_{\gamma}(Z)$ for $i\in\{0,1,\ldots, n-1\}$. This implies that each call of the function \textsc{Evaluate Derivatives A} performs $4^n\cdot p^b\cdot \poly(a,b,n)$ $\F_q$-operations. 
 
Now we return to analyzing the cost taken by each iteration of the \emph{for} loop at Line $8$ 
in Algorithm \autoref{algo:polynomial-evaluation-v2}.
From the above discussion, for each $i\in[N]$, the set $E_i$ can be computed in $4^n\cdot p^b\cdot \poly(a,b,n)$ $\F_q$-operations.  Given $E_i$, applying  \autoref{lem:hermitian-interpolation}, the interpolation of  $h_i(t)$ requires $\poly(a, b, d,n)$ operations in $\F_q$. Thus, each iteration of the \emph{for} loop at Line $8$ 
in Algorithm \autoref{algo:polynomial-evaluation-v2}
takes $4^n\cdot p^b\cdot \poly(a,b,d,n)$ $\F_q$-operations. Therefore, the total cost of the \emph{for} loop is $$N\cdot 4^n\cdot p^b\cdot \poly(a,b,d,n)$$ $\F_q$-operations. Since $ad<p^b\leq adp$, combining the complexities of all the components, we get that Algorithm \autoref{algo:polynomial-evaluation-v2} performs $$(N+(adp)^n)\cdot 4^n\cdot \poly(a,d,n,p)$$
$\F_q$-operations.
\end{proof}

\section{Multipoint evaluation with improved field dependence} \label{sec: mme large fields}
In this section, we build on the ideas in Algorithm \autoref{thm:polynomial-evaluation-v2} to improve the dependence on the field size. Our main theorem, which is a formal statement of our main result \autoref{thm:mme informal} stated in the introduction.  
\begin{theorem}
\label{thm:polynomial-evaluation-v4}
Let $p$ be a prime and $q=p^a$ for some positive integer $a$. There is a deterministic algorithm such  that on input an $n$-variate polynomial $f(\var x)$ over $\F_q$ with individual degree less than $d$, points  $\pmb\alpha_1,\pmb\alpha_2,\ldots, \pmb\alpha_N$ from $\F_q^n$ and a non-negative integer $\ell\leq\log^\star_p(a)$, it outputs  $f(\pmb\alpha_i)$ for all $i\in[N]$ in time 

$$\left(N\cdot\left(2dp\log_p(dp)\right)^\ell+\left(2rdp\log_p(dp)\right)^n\right)\cdot O(\ell+1)^n\cdot \poly(a,d,n,p) \, ,$$
where  $r=\max\{2, \log^{\circ\ell}_p(a)\}$.


\end{theorem}

\subsection{A description of the algorithm}
We start by describing the algorithm, followed by its analysis. 

\begin{algorithm}[H]
\caption{Efficient multivariate polynomial evaluation over large fields}
	\label{algo:polynomial-evaluation-v4}
	\textbf{Input:} An $n$-variate polynomial $f(\var x)\in\F_q[\var x]$ with individual degree less than $d$, $N$ points $\pmb\alpha_1,\pmb\alpha_2,\ldots,\pmb\alpha_N$ from $\F_q^n$, and a non-negative integer $\ell\leq \log_p^\star(a)$ where $q=p^a$ and $p$ is the characteristic of $\F_q$.\\
	\textbf{Output:} $f(\pmb\alpha_1),\ldots,f(\pmb\alpha_N)$.
	\begin{algorithmic}[1]
		\State Let $v_0(Y_0)$ be an irreducible polynomial in $\F_p[Y_0]$ of degree $a$ and $\F_q=\F_p[Y_0]/\ideal{v_0(Y_0)}$.
		\State $\pts_0\leftarrow \{\pmb\alpha_i\,\mid\, i\in[N]\}$,  $a_0\leftarrow a$, and $q_0\leftarrow p^{a}$.
        \State \Call{Polynomial Evaluation}{$0$}.   { $\quad$     (Recursive call)}
        \State Output $\eval_{0,\mathbf 0}$.
         \Function{Polynomial Evaluation}{$i$}
               \State Let $a_{i+1}$ be the smallest positive integer such that $p^{a_{i+1}}>a_{i}d$, and $q_{i+1}\leftarrow q^{a_{i+1}}$.
               \State Compute an irreducible polynomial $v_{i+1}(Y_{i+1})$ over $\F_q$ of degree $a_{i+1}$ and $$\F_{q_{i+1}}=\F_q[Y_{i+1}]/\ideal{v_{i+1}(Y_{i+1})}. \text{ (\autoref{lem:finite-field-construction})}$$
               \State Compute the subfield $\F_{p^{a_{i+1}}}$ of $\F_{q_{i+1}}$. (\autoref{lem:subfield-construction})
               \State Compute an element $\beta_{i}$ in $\F_{q_i}$ s.t. $\{1,\beta_i, \ldots,\beta_i^{a_i-1}\}$ forms an $\F_p$-basis for $\F_{p^{a_i}}$. (\autoref{lem:subfield-construction})
               \State $\pts_{i+1}\leftarrow \emptyset$.
               \For{all $\pmb\alpha\in \pts_i$}
                   \State Let $\pmb\alpha=\pmb\alpha_{0}+\pmb\alpha_{1}\beta_i+\cdots+\pmb\alpha_{a_i-1}\beta_{i}^{a_i-1}$, where $\pmb\alpha_{j}\in \F_{p}^n$. 
                   \State Compute $\pmb\alpha_0,\ldots,\pmb\alpha_{a_i-1}$. (\autoref{lem:subfield-construction})
                   \State Let $\var g_{\pmb\alpha}(t)$ be the curve defined as $\pmb\alpha_{0}+ \pmb\alpha_{1}t+ \cdots+ \pmb\alpha_{a_i-1}t^{a_i-1}$.
                   \State $P_{\pmb\alpha}\leftarrow \{\var g_{\pmb\alpha}(\gamma)\,\mid\,\gamma\in \F_{p^{a_{i+1}}}\}$ (\autoref{lem:univariate-evaluation}), and $\pts_{i+1}\leftarrow \pts_{i+1}\cup P_{\pmb\alpha}$.\label{label:algo-v4-pts}
               \EndFor
               \If{$i<\ell$}
                  \State \Call{Polynomial Evaluation}{$i+1$}.
               \Else
                  \State Compute all the polynomials in $\hpartial^{\leq (\ell+1)(n-1)}(f)$. (\autoref{lem:computing-hasse-derivation})
                  \State Evaluate all the polynomials in $\hpartial^{\leq (\ell+1)(n-1)}(f)$ over the grid $\F_{p^{a_{\ell+1}}}^n$. (\autoref{lem: multidim FFT})
                  \State Observe that $\pts_{\ell+1}$ is a subset of $\F_{p^{a_{\ell+1}}}^n$.
                  \State For all $\var e\in\N^n$ with $|\var e|_1\leq (\ell+1)(n-1)$, $\eval_{\ell+1, \var e}=\{(\pmb\alpha, \hpartial_{\var e}(f)(\pmb\alpha))\,\mid\, \pmb\alpha\in\F_{p^{a_{\ell+1}}}^n\}$.
               \EndIf
               \For{all $\var e\in\N^n$ s.t. $|\var e|_1\leq i(n-1)$}
               \State $\eval_{i,\var e}\leftarrow \emptyset$.
               \For{all $\pmb\alpha\in \pts_i$}
                   \State Let $h_{\var e,\pmb\alpha}(t)=\hpartial_{\var e}(f)(\var g_{\pmb\alpha }(t))$.
                   \State Let $E_{\var e,\pmb\alpha}=\{(\gamma, \der{h_{\var e,\pmb\alpha}}{0}(\gamma), \ldots, \der{h_{\var e,\pmb\alpha}}{n-1}(\gamma))\,\mid\,\gamma\in\F_{p^{a_{i+1}}}\}$.
                   \State Using \textsc{Evaluate Derivatives B} 
                   with input $(\var g_{\pmb\alpha}(t), i, \var e)$, compute $E_{\var e,\pmb\alpha}$. 
                   \State Using $E_{\var e,\pmb\alpha}$, interpolate $h_{\var e,\pmb\alpha}(t)$. (\autoref{lem:hermitian-interpolation})
                   \State $\eval_{i,\var e}\leftarrow\eval_{i, \var e}\cup\{(\pmb\alpha, h_{\var e,\pmb\alpha}(\beta_i))\}$. (\autoref{lem:univariate-evaluation})
               \EndFor
               \EndFor
        \EndFunction
	\end{algorithmic}
\end{algorithm}

\begin{algorithm}[H]
	\caption{Evaluating Hasse derivatives for Algorithm \ref{algo:polynomial-evaluation-v4}}
	\label{algo:computing-hasse-derivative-v3}
	\begin{algorithmic}[1]
               \Function{Evaluate Derivatives B }{$\var g(t)$, $k$, $\var e$}
                   \State Let $\var g(t)=(g_1,\ldots,g_n)$.
                   \State Let $\var e=(e_1,\ldots,e_n)$.
                   \State Let $g_i(t+Z)=g_i(t)+Z\tilde g_i(t,Z)$, for all $i\in[n]$, and $\tilde{\var g}(t,Z)=(\tilde g_1(t,Z),\ldots, \tilde g_n(t,Z))$.
                   \State Compute $\tilde g_i(t,Z)$ for all $i\in[n]$. (\autoref{lem:computing-hasse-derivation})
	               \State For all $\var b=(b_1,\ldots, b_n)\in\N^n$, let $\tilde{\var g}_{\var b}=\prod_{i=1}^n\tilde g_i^{b_i}$.
	               \State Compute the set of polynomials $\{\tilde{\var g}_{\var b}(t,Z)\,\mid\,|\var b|_1<n\}$. (Polynomial multiplication)
	               \State Let $D$ be an $((k+1)(n-1)+1)\times n$ array such that, $$D_{i,j}=\binom{j}{i}\,\mathrm{mod}\, p,\ \text{where } i\in\{0,1,\ldots,n-1\},j\in\{0,1,\ldots,(k+1)(n-1)\}$$
	               \State Like Algorithm \ref{algo:computing-hasse-derivative}, we can compute $D$ using $\F_p$-operations. 
	               \State For $\var b=(b_1,\ldots, b_n)\in\N^n$ such that $|\var b|_1<n$, $$c_{\var b}\leftarrow \prod_{i=1}^nD_{e_i+b_i, b_i}.$$\label{label:v4-binomial-coeff}
	               \State $P\leftarrow \emptyset$.
		           \For{all $\gamma \in\F_{p^{a_{k+1}}}$}
		               \State Using evaluations of polynomials in $\hpartial^{\leq (k+1)(n-1)}(f)$ over $\pts_{k+1}$, compute $$h_\gamma(Z)=\sum_{i=0}^{n-1}Z^i\sum_{\var b\in\N^n:|\var b|_1=i}c_{\var b}\hpartial_{\var e+\var b}(f)(\var g(\gamma))\tilde{\var g}_{\var b}(\gamma, Z).$$
		                \State For all $i\in\{0,1,\ldots, n-1\}$, extract $\coeff_{Z^{i}}(h_{\gamma})$.
		                \State $P\leftarrow P\cup\{(\gamma, \coeff_{Z^0}(h_\gamma),\coeff_{Z^1}(h_\gamma),\ldots, \coeff_{Z^{n-1}}(h_\gamma))\}$.
		           \EndFor
		           \State \Return $P$.
               \EndFunction
	\end{algorithmic}
\end{algorithm}

\subsection{Analysis of Algorithm \autoref{algo:polynomial-evaluation-v4}}

\paragraph*{A useful lemma. }
The following lemma would be useful for the time complexity analysis of Algorithm \autoref{algo:polynomial-evaluation-v4}. 

\begin{lemma}
\label{lem:number-sequence}
Let $a,d$ and $p$ be three positive integers such that $p,d$ are greater than $1$. Let $(a_0,a_1,a_2,\ldots)$ be a sequence of positive integers with the following properties: $a_0=a$, and for all $i>0$, $a_i$ be the smallest positive integer such that $p^{a_i}> da_{i-1}$. Then for all non-negative integer $i$, $$a_i\leq2r_i\log_p(dp),$$ where $r_i=\max\{2,\log^{\circ i}_p(a)\}.$ Furthermore, for any non-negative integer $\ell\leq \log^\star_p(a)$, $$\prod_{i=0}^\ell a_i\leq (2\log_p(dp))^\ell\cdot a^{1+o(1)}.$$
\end{lemma}

\begin{proof}
From the definition of the sequence, it is not hard to see that for every positive integer $i$, $p^{a_i}\leq dpa_{i-1}$. We inductively show that for every non-negative integer $i$, $a_i\leq 2r_i\log_p(dp)$. It is true for $i=0$ since $\log^{\circ i}_p(a)=a$ for $i=0$ and $\log_p(dp)\geq 1$. This establishes our base case. Now assume that  $a_i\leq 2r_i\log_p(dp)$ for some integer $i\geq 0$. We know that $p^{a_{i+1}}\leq dpa_{i}$. Taking logarithm with respect to $p$, we get that $a_{i+1}\leq \log_p(dp)+\log_p(a_i)$. From the induction hypothesis, we get that 
\begin{align*}
a_{i+1}&\leq \log_p(dp)+\log_p(2r_i\log_p(dp))\\
&=\log_p(dp)+\log_p\log_p(dp)^2+\log_p(r_i)
\end{align*}
From the definition of $r_{i+1}$, we get that $r_{i+1}\geq \log_p(r_i)$. Also, $(dp)^2\leq p^{dp}$ for all $d,p\geq 2$. Therefore, $a_{i+1}\leq 2\log_p(dp)+r_{i+1}$. Since both $2\log_p(dp)$ and $r_{i+1}$ are $\geq 2$,  $$2\log_p(dp)+r_{i+1}\leq 2r_{i+1}\log_p(dp).$$ This completes the induction step.

Now we prove the second part of the above lemma. From the first half, we get that $$\prod_{i=0}^\ell a_i\leq a(2\log_p(dp))^\ell\cdot \prod_{i=1}^\ell r_i.$$ Let $k$ be the largest non-negative integer such that $\log^{\circ k}_p(a)\geq 2$. First, assume that $\ell\leq k$.
Then
\begin{align*}
\prod_{i=0}^\ell a_i &\leq a(2\log_p(dp))^\ell\cdot \prod_{i=1}^\ell r_i \\
&\leq a(2\log_p(dp))^\ell\cdot \prod_{i=1}^\ell\log_p^{\circ i}(a)\\
&\leq (2\log_p(dp))^\ell\cdot a^{1+o(1)}.
\end{align*}
 Since $\log_p^\star(a)$ can be at most $k+2$, $\ell\leq k+2$. Now, assume that $\ell>k$.
\begin{align*}
\prod_{i=0}^\ell a_i &\leq a(2\log_p(dp))^\ell\cdot \prod_{i=1}^\ell r_i\\ &\leq a(2\log_p(dp))^\ell\cdot\left(\prod_{i=1}^{k}\log^{\circ i}_p(a)\right)\cdot 2^{\ell-k}\\
&\leq (2\log_p(dp))^\ell\cdot a^{1+o(1)}.
\end{align*}
\end{proof}

We are now ready to discuss the proof of \autoref{thm:polynomial-evaluation-v4}. As discussed in \autoref{sec: proof overview}, the main idea in reducing the dependence of the running time on the underlying field is to reduce the question of multipoint evaluation over the field $\F_{p^{a_0}}$ to an instance of multipoint evaluation with a larger number of points, but all these points lie in a smaller field $\F_{p^{a_1}}$, where $p^{a_1} \geq a_0dn$. This reduction in the size  leads to a significant decrease in the degree of the curves used in the local computation step, at the cost of increasing the number of points. We now make this intuition formal, and prove the necessary quantitative bounds. 

\begin{proof}[Proof of \autoref{thm:polynomial-evaluation-v4}]
We start with a proof of correctness. 
\paragraph*{Correctness of Algorithm \autoref{algo:polynomial-evaluation-v4}. }

We prove that Algorithm \autoref{algo:polynomial-evaluation-v4} computes $f(\pmb\alpha_i)$ for all $i\in[N]$ with the desired time complexity. First, we briefly highlight the main difference between Algorithm \autoref{algo:polynomial-evaluation-v4} and Algorithm \ref{algo:polynomial-evaluation-v2}. In Algorithm \autoref{algo:polynomial-evaluation-v2}, we construct the field $\F_{q^b}$, a degree $b$ extension of $\F_q$, such that both $\F_q$ and $\F_{p^b}$ are its subfields and $p^b> ad$. Next, we evaluate all the polynomials in $\hpartial^{\leq n-1}(f)$ over the grid $\F_{p^b}^n$. Finally, we reduce the evaluation of $f(\var x)$ at points in $\F_q^n$ to the evaluation of the polynomials in $\hpartial^{\leq n-1}(f)$ at the points in $\F_{p^b}^n$, and use the evaluations of $\hpartial^{\leq n-1}(f)$ at the grid $\F_{p^b}^n$ to compute $f(\pmb\alpha_i)$ for all $i\in[N]$. In short, Algorithm \autoref{algo:polynomial-evaluation-v2} reduces the evaluation of $f(\var x)$ at points in $\F_q^n$ to the evaluation of $\hpartial^{\leq n-1}(f)$  at the points in a "smaller" grid $\F_{p^b}^n$. In Algorithm \autoref{algo:polynomial-evaluation-v4}, we repeat this reduction $(\ell+1)$ times where at the $i$th iteration, we reduce the question of evaluating the set of all Hasse derivatives of $f$ of order up to $i(n-1)$ on a subset of points($\pts_i$) in $\F_{p^{a_i}}^n$ to the question of evaluating all the Hasse derivatives of $f$ of order up to $(i+1)(n-1)$ on a subset of points($\pts_{i+1}$) in $\F_{p^{a_{i+1}}}^n$. Finally, at the $\ell$th iteration, we reach a much "smaller" grid $\F_{p^{a_{\ell+1}}}^n$ (compare to $\F_q^n$) where we evaluate all the Hasse derivatives of $f(\var x)$ up to order  $(\ell+1)(n-1)$. We now show via an induction on $i$, with $i$ decreasing from $\ell$ to $0$, the following claim holds. 

\begin{claim}\label{claim: inductive claim in correctness proof}
For every $i \in \{0, 1, \ldots, \ell\}$, at the end of the function call \textsc{Polynomial Evaluation $(i)$}, we have correctly computed the evaluation of all polynomials in the set $\hpartial^{\leq i(n-1)}(f)$ at all points in the set $\pts_{i}$.    
\end{claim}
 Recall that the set $\pts_{0}=\{\pmb\alpha_i\,\mid\,i\in[N]\}$ is the original set of input points and hence, this suffices for the correctness of the algorithm. To proceed with the induction, we need the following subclaim whose proof we defer to the end. Recall the set $E_{\vece, \pmb\alpha}$ defined in Line $27$ of Algorithm \autoref{algo:polynomial-evaluation-v4}. 
\begin{claim}
\label{claim:correctness-derivatives}
Given $(\var g_{\pmb\alpha}(t), i, \var e)$ as the input, the function \textsc{Evaluate Derivatives B} of Algorithm \ref{algo:computing-hasse-derivative-v3} computes the set $E_{\var e,\pmb\alpha}$ in time  $$4^n\cdot \poly(a_i, a_{i+1}, d,n,p) \, .$$ 
\end{claim}
We now proceed with the inductive proof of \autoref{claim: inductive claim in correctness proof}. 
\paragraph*{Base case. }For $i = \ell$,  in Line $22$ of Algorithm \autoref{algo:polynomial-evaluation-v4}, for every $\vece \in \N^{n}$ with $|\vece|_1 \leq (\ell + 1)(n-1)$, we first compute the set $\eval_{\ell+1, \ve}$ which is the evaluation table of the polynomial $\hpartial_{\vece}(f)$ at all points in the set $\F_{p^{a_{\ell + 1}}}^n$ via a multidimensional FFT (\autoref{lem: multidim FFT}). We now claim that at the end of the subsequent  \emph{for} loop (Lines $23-30$), for every $\pmb\alpha \in \pts_{\ell}$ and for every $\vece$ in $\N^n$ with $|\vece|_1 \leq \ell(n-1)$, we have the value of $\hpartial_{\vece}(f)$ at $\pmb\alpha$. 
For this, consider a point $\pmb\alpha \in \pts_{\ell}$ and  $\var e\in\N^n$ such that $|\var e|_1\leq \ell(n-1)$. Now the curve $\var g_{\pmb\alpha}(t)$ can be computed efficiently as in the proofs of correctness of Algorithm \autoref{algo:polynomial-evaluation-v1} and Algorithm \autoref{algo:polynomial-evaluation-v2}.  Then, $h_{\var e, \pmb\alpha}=\hpartial_{\var e}(f)(\var g_{\pmb\alpha}(t))$ is a  polynomial of degree less than $a_{\ell}dn$. From \autoref{claim:correctness-derivatives}, we have that the function \textsc{Evaluate Derivatives B} correctly computes the set $E_{\vece, \pmb\alpha}$. Moreover, by our choice of $a_{i}'s$, we have that $p^{a_{\ell+1}}\geq a_{\ell}d$. Thus,  from \autoref{lem:hermitian-interpolation}, we can interpolate $h_{\var e, \pmb\alpha}(t)$ correctly from the set $E_{\var e,\pmb\alpha}$.   

\paragraph*{Induction step. }
In the induction step,  we assume \autoref{claim: inductive claim in correctness proof} is true for $i = i_0 \leq \ell $ and show that it holds for the iteration $i_{0}-1$. The proof of this is precisely the same as that of the base case. The only difference is that in the base case, the set $\eval_{\ell + 1, \vece}$ for every $\vece \in \N^n$ with $|\vece|_1 \leq (\ell + 1)(n-1)$ was computed in Line $22$ of the algorithm directly via the multidimensional FFT. In the induction step, for iteration $i_0-1$, we need the corresponding set  $\eval_{i_{0}, \vece}$ for every $\vece \in \N^n$ with $|\vece|_1 \leq i_0 (n-1)$. Note that these are sets of evaluations of all derivatives of order at most $i_0(n-1)$ of $f$ on the point set $\pts_{i_0}$ that are guaranteed to be available to us by the induction hypothesis. The rest of the argument is exactly as that in the base case. We skip the details.

A point to note for the proof of both the base case and the induction step of \autoref{claim: inductive claim in correctness proof}, is that  $h_{\var e, \pmb\alpha}$ is a polynomial with coefficients in $\F_q$ and $\F_q$ is a subfield of  $\F_{q_i}$ for every $i \in \{0, 1, \ldots, \ell \}$, so all the algebra is consistent here. 

This completes the proof of \autoref{claim: inductive claim in correctness proof} and hence the correctness of \autoref{algo:polynomial-evaluation-v4} (which follows from $i = 0$ case of the claim), modulo the proof of \autoref{claim:correctness-derivatives}. We defer that to the end of this section, and discuss the time complexity of Algorithm \autoref{algo:polynomial-evaluation-v4}.

\paragraph*{Time complexity of Algorithm \autoref{algo:polynomial-evaluation-v4}. }

Now we discuss the time complexity of Algorithm \ref{algo:polynomial-evaluation-v4}. We  rely on the following claim, whose proof we defer to the end of this section.
\begin{claim}
\label{claim:v4-points-size}
$|\pts_{\ell}|\leq N\cdot(2dp\log_p(dp))^\ell\cdot a^{1+o(1)}.$
\end{claim}
For all $i\in\{0,1,\ldots,\ell\}$, let $T_i$ be the time complexity of the $i$th invocation of the function \textsc{Polynomial Evaluation}.  Next, we discuss the various components of $T_i$.
\begin{enumerate}
    \item Using \autoref{lem:finite-field-construction}, the field $\F_{q_{i+1}}=\F_q[Y_{i+1}]/\ideal{v_{i+1}(Y_{i+1})}$  can be constructed in $\poly(a,a_{i+1},p)$ $\F_q$-operations. From \autoref{lem:subfield-construction}, the cost of computing all the elements of the subfield $\F_{p^{a_{i+1}}}$ (of $\F_{q_{i+1}}$) is $p^{a_{i+1}}\cdot \poly(a,a_{i+1},p)$ $\F_q$-operations.
    
    \item Using \autoref{lem:subfield-construction}, we can also compute the element $\beta_{i}$ in $\poly(a, a_i, p)$ $\F_q$-operations. In addition, \autoref{lem:subfield-construction} ensures that for every $\gamma\in\F_{p^{a_i}}$, the $\F_p$-linear combination of $\gamma$ with respect to $\{1,\beta_i,\ldots,\beta_i^{a_i-1}\}$ can be computed using $\poly(a_i)$ $\F_q$-operations. Therefore, for every $\pmb\alpha\in\pts_i$, the cost of computing the curve $\var g_{\pmb\alpha}(t)$   is $\poly(a_i,n)$ $\F_q$-operations. This implies that the set $P_{\pmb\alpha}$ can be constructed in $p^{a_{i+1}}\cdot \poly(a_i,a_{i+1},n)$ $\F_q$-operations. Thus, the total cost of computing the set $\pts_{i+1}$ is $$|\pts_i|\cdot p^{a_{i+1}}\cdot \poly(a_i,a_{i+1},n)$$ $\F_q$-operations.
    
    \item For $i\in\{0,1,\ldots,\ell-1\}$, we need to add the cost of $(i+1)$th call of \textsc{Polynomial Evaluation}, that is $T_{i+1}$. For $i=\ell$, using \autoref{lem:computing-hasse-derivation}, we can compute the set $\hpartial^{\leq (\ell+1)(n-1)}(f)$ in $\binom{(\ell+1)(n-1)+n}{n}\cdot d^n\cdot \poly(n)$ many $\F_q$-operations. Since $\binom{(\ell+1)(n-1)+n}{n}\leq O(\ell+1)^n$, the total cost of computing the set $\hpartial^{\leq (\ell+1)(n-1)}(f)$ is $$O(\ell+1)^n\cdot d^n\cdot \poly(n)$$ $\F_q$-operations.  We have to evaluate all the polynomials in $\hpartial^{\leq(\ell+1)(n-1)}(f)$ over the grid $\F_{p^{a_{\ell+1}}}^n$. Using \autoref{lem: multidim FFT}, each polynomial in $\hpartial^{\leq(\ell+1)(n-1)}(f)$ can be evaluated over the grid in $$(d^n+p^{na_{\ell+1}})\cdot \poly(a, a_{\ell+1}, d,n,p)$$ many $\F_q$-operations. Thus, the total cost of evaluating all the polynomials in $\hpartial^{\leq(\ell+1)(n-1)}(f)$ over the grid is $$(a_\ell d p)^n\cdot O(\ell+1)^n\cdot\poly(a, a_{\ell+1}, d,n,p)$$ $\F_q$-operations since $p^{a_{\ell+1}}\leq a_\ell dp$.
    
    \item We have to compute the set $\eval_{i,\var e}$ for all $\var e\in\N^n$ with $|\var e|_1\leq i(n-1)$. Let $\var e\in\N^n$ with $|\var e|_1\leq i(n-1)$, and $\pmb\alpha\in\pts_i$. Then, using \autoref{claim:correctness-derivatives}, the set $E_{\var e,\pmb\alpha}$ can be computed in $4^n\cdot \poly(a_i, a_{i+1}, d, n, p)$ many $\F_q$-operations. Applying \autoref{lem:hermitian-interpolation}, the polynomial $h_{\var e,\pmb\alpha}(t)$ can be interpolated from $E_{\var e,\pmb\alpha}$  using $\poly(a_{i},a_{i+1}, d,n)$ $\F_q$-operations. Finally, $h_{\var e,\pmb\alpha}(\beta_i)$ can be computed in $\poly(a_i,d,n)$ $\F_q$-operations. Thus, for an $\var e\in \N^n$ with $|\var e|_1\leq i(n-1)$, the cost of computing $\eval_{i,\var e}$ is $|\pts_i|\cdot 4^n\cdot (a_i,a_{i+1},d,n,p)$ $\F_q$-operations. The number of $\var e\in\N^n$ with $|\var e|_1\leq i(n-1)$ is $\binom{i(n-1)+n}{n}$, and it is upper bounded by $O(i+1)^n$. Therefore, the total of computing $\eval_{i,\var e}$ for all $\var e\in\N^n$ with $|\var e|_1\leq i(n-1)$ is $$|\pts_i|\cdot O(i+1)^n\cdot \poly(a_i, a_{i+1}, d,n,p)$$ $\F_q$-operations.  
    \end{enumerate}
    The choice of $a_{i+1}$ gives us that $p^{a_{i+1}}\leq a_idp$. From \autoref{lem:number-sequence}, $a_i\leq 4a\log_p(dp)$ for all $i\geq 1$. Thus, from the above discussion, for all $i\in\{0,1,\ldots,\ell-1\}$, $$T_{i}\leq T_{i+1}+ |\pts_i|\cdot O(i+1)^n\cdot \poly(a,d,n,p).$$ Also, the complexity of the $\ell$th invocation of \textsc{Polynomial Evaluation} is $$T_{\ell}\leq (a_\ell d p)^n\cdot O(\ell+1)^n\cdot \poly(a,d,n,p).$$ Therefore, the overall $\F_q$-operations performed by Algorithm \ref{algo:polynomial-evaluation-v4} is 
\begin{align*}
    T_0 &\leq \sum_{i=0}^{\ell-1}|\pts_i|\cdot O(i+1)^n\cdot \poly(a,d,n,p)+T_\ell\\
    &\leq \ell|\pts_\ell|\cdot O(\ell+1)^n\cdot \poly(a,d,n,p)+(a_\ell dp)^n\cdot O(\ell+1)^n\cdot\poly(a,d,n,p)\\
    &\leq (|\pts_\ell|+(a_\ell dp)^n)\cdot O(\ell+1)^n\cdot \poly(a,d,n,p).
\end{align*}
Using \autoref{lem:number-sequence},  $a_\ell\leq 2r\log_p(dp)$, where $r=\max\{2,\log_p^{\circ \ell}(a)\}$. From \autoref{claim:v4-points-size}, $$|\pts_\ell|\leq N\cdot(2dp\log_p(dp))^\ell\cdot a^{1+o(1)}.$$ Therefore, the number of $\F_{q}$-operations performed by Algorithm \ref{algo:polynomial-evaluation-v4} $$T_0\leq \left(N\cdot \left(2dp\log_p(dp)\right)^\ell+\left(2rdp\log_p(dp)\right)^n\right)\cdot O(\ell+1)^n\cdot \poly(a,d,n,p).$$   
\end{proof} 
This completes the proof of \autoref{thm:polynomial-evaluation-v4} modulo the proofs of the \autoref{claim:correctness-derivatives} and \autoref{claim:v4-points-size}. We now discuss these missing proofs.
\paragraph{Proof of \autoref{claim:correctness-derivatives}}
\begin{proof}[Proof of \autoref{claim:correctness-derivatives}]
According to the function call, $\var g(t)=\var g_{\pmb\alpha}(t)$ and $k=i$. Also, $\var g(t)=(g_1,\dots,g_n)$, $g_i(t+Z)=g_i(t)+Z\tilde g_i(t,Z)$ for all $i\in[n]$ and for all $\var b=(b_1,\ldots,b_n)\in\N^n$, $\tilde{\var g}_{\var b}=\prod_{i=1}^n\tilde g_i^{b_i}$. Let $\hpartial_{\var e}(f)=f_{\var e}$. Let $$h(t+Z)=\sum_{i=0}^{n-1}Z^i\sum_{\var b\in\N^n:|\var b|_1\leq n-1}\hpartial_{\var b}(f_{\var e})(\var g(t))\cdot \tilde{\var g}_{\var b}.$$ Then, from \autoref{lem:hasse-derivative-property}, $$h(t+Z)=\sum_{i=0}^{n-1}Z^i\sum_{\var b\in\N^n:|\var b|_1\leq n-1}\binom{\var e+\var b}{\var b}\hpartial_{\var e+\var b}(f)(\var g(t))\cdot \tilde{\var g}_{\var b}.$$ In step $10$ of Algorithm \ref{algo:computing-hasse-derivative-v3}, for all $\var b\in\N^n$ with $|\var b|_1\leq n-1$, $c_{\var b}=\binom{\var e+\var b}{\var b}.$ Therefore, $$h(t+Z)=\sum_{i=0}^{n-1}Z^i\sum_{\var b\in\N^n:|\var b|_1\leq n-1}c_{\var b}\hpartial_{\var e+\var b}(f)(\var g(t))\cdot \tilde{\var g}_{\var b}.$$ Applying \autoref{lem:hasse-derivative}, we get that for all $i\in\{0,1,\ldots,n-1\}$, the $i$th order Hasse derivative of $f_{\var e}(\var g(t))$ is same as $\coeff_{Z^i}(h(t+Z))$. Hence, for all $\gamma\in\F_{p^{a_{k+1}}}$, the evaluation of $i$th order Hasse derivative of $f_{\var e}(\var g(t))$ at $\gamma$ is equal to $\coeff_{Z^i}(h(\gamma+Z))$. In terms of notation used in Algorithm \ref{algo:computing-hasse-derivative-v3}, $\coeff_{Z^i}(h(\gamma+Z))$ is same as $\coeff_{Z^i}(h_{\gamma}(Z))$. This implies that the evaluation of $i$th order Hasse derivative of $f_{\var e}(\var g(t))$ at $\gamma$ is equal to $\coeff_{Z^i}(h_{\gamma}(Z))$. This implies that given $(\var g_{\pmb\alpha}(t), k, \var e)$ as input, the set $P$ returned by \textsc{Evaluate Derivatives B} is same as $E_{\var e, \pmb\alpha}$. Now,  to compute $h_{\gamma}(Z)$, we need access to $\tilde{\var g}_{\var b}(\gamma, Z)$ for all $\var b\in\N^n$ with $|\var b|_1<n$. \autoref{lem:computing-hasse-derivation} ensures that we can compute $\tilde g_i(t,Z)$ for all $i\in[N]$. After we have all $\tilde g_i(t,Z)$'s, we can compute $\tilde{\var g}_{\var b}(t,Z)$ and evaluate it at $t=\gamma$. Also, we need the access of $\hpartial_{\var e+\var b}(f)(\var g(\gamma))$ for all $\var b\in\N^n$ with $|\var b|_1<n$. Observe that $|\var e+\var b|_1\leq (k+1)(n-1)$ and $\var g(\gamma)\in\pts_{k+1}$. Therefore, from the evaluations of the polynomials $\hpartial^{\leq (k+1)(n-1)}(f)$ at points $\pts_{k+1}$, we get $\hpartial_{\var e+\var b}(f)(\var g(\gamma))$.

Now we discuss the number of $\F_q$-operations performed by Algorithm \ref{algo:computing-hasse-derivative-v3}. From \autoref{lem:computing-hasse-derivation}, we can compute $\tilde g_i(t,Z)$ for all $i\in[n]$ in $\poly(a_k, n)$ $\F_q$-operations. Since each $\tilde g_i(t,Z)$ is a bivariate polynomial of individual degree less than $a_k$ and $\binom{n+n-1}{n}\leq 4^n$, we can compute the set $\{\tilde{\var g}_{\var b}(t,Z)\,\mid\, |\var b|_1<n\}$ in  $4^n\cdot \poly(a_k,n)$ $\F_q$-operations. Computing all $c_{\var b}$'s takes $4^n\cdot \poly(n)+\poly(k,n)$ $\F_q$-operations. Given a $\gamma\in\F_{p^{a_{k+1}}}$, from \autoref{lem:subfield-construction}, we can evaluate $\tilde{\var g}_{\var b}$ at $t=\gamma$ in $\poly(a_k, a_{k+1}, n)$ $\F_q$-operations. Thus, for any $\gamma\in\F_{p^{a_{k+1}}}$, the cost of computing $h_{\gamma}(Z)$ is $4^n\cdot \poly(a_k, a_{k+1}, n)$. Once we have $h_\gamma(Z)$ as its list of coefficients, we collect the coefficients of $Z^i$ for all $i\in\{0,1,\ldots,n-1\}$. From the choice of $a_{k+1}$, $p^{a_{k+1}}\leq a_kdp$. Thus, the total cost of Algorithm \ref{algo:computing-hasse-derivative-v3} is $$4^n\cdot \poly(a_k, a_{k+1}, d, p, n)$$ $\F_q$-operations.
\end{proof}
\vspace{1cm}
\paragraph*{Proof of \autoref{claim:v4-points-size}}
\begin{proof}[Proof of \autoref{claim:v4-points-size}]
First, we show that  for all $i\in\{0,1,\ldots,\ell-1\}$, $|\pts_{i+1}|\leq |\pts_i|\cdot (a_idp)$. From the step $15$ of Algorithm \ref{algo:polynomial-evaluation-v4}, the set $$\pts_{i+1}=\cup_{\pmb\alpha\in\pts_i}P_{\pmb\alpha}.$$ The definition of $P_{\pmb\alpha}$ ensures that its size is at most $p^{a_{i+1}}$, which is upper bounded by $a_{i}dp$. Therefore, the size of $\pts_{i+1}$ is at most $|\pts_i|\cdot (a_idp)$. This implies that $$\pts_{\ell}\leq N\cdot (dp)^\ell\cdot \prod_{i=0}^\ell a_i.$$ From \autoref{lem:number-sequence}, $\prod_{i=0}^\ell a_i\leq (2\log_p(dp))^\ell\cdot a^{1+o(1)}.$ Therefore, $$|\pts_\ell|\leq N\cdot (2dp\log_p(dp))^\ell\cdot a^{1+o(1)}.$$ 
\end{proof}

\section{An algebraic data structure for polynomial evaluation}\label{sec:data structures full}
In this section, we discuss the implication of our multipoint evaluation algorithms to the question of data structures for polynomial evaluation. For functions $s: \N \to \N$ and $t: \N \to \N$, an algebraic data structure for univariate polynomial evaluation over a field $\F$ with space complexity $s(n)$ and time complexity $t(n)$ is specified by a preprocessing map and a query algorithm. For every $n \in \N$, the preprocessing map maps a polynomial $f \in \F[X]$ of degree less than $n$ to an $s(n)$ dimensional vector over $\F$ which we denote by ${\cal D}_f$ and the query algorithm is an algebraic algorithm that on any input $\alpha \in \F$, accesses at most $t(n)$ coordinates of ${\cal D}_f$ and correctly outputs $f(\alpha)$. 

For this discussion $\F$ is a finite field, and as mentioned in \autoref{rmk: description of the finite field}, we assume that the query algorithm has access to a description of the field, for instance via an irreducible polynomial of appropriate degree over the base field. 
Now, we formally state the main result for this section. 
\begin{theorem}\label{thm: data structure technical}
    Let $p$ be a fixed prime. Then, for all sufficiently large $n \in \N$ and all fields $\F_{p^a}$ with $a = \poly(\log n)$, there is an algebraic data structure for polynomial evaluation for univariate polynomials of degree less than $n$ over $\F_{p^a}$ that has space complexity at most $n^{1 + o(1)}$ and query complexity at most $n^{o(1)}$. 
    
    Moreover, given a description of $\F_{p^a}$ (via an irreducible polynomial of degree $a$ over $\F_p$ ) there is an algebraic algorithm that when given the coefficients of a univariate polynomial $f$ of degree less than $n$ computes the output of the preprocessing map on $f$, denoted here by ${\cal D}_f$ in time $n^{1 + o(1)}$ and an algebraic algorithm which, when given an $\alpha \in \F_{p^a}$ and  ${\cal D}_f$, outputs $f(\alpha)$ in time $n^{o(1)}$. 
\end{theorem}

We recall that for fields of small characteristic and size $\poly(n)$, \autoref{thm: data structure technical} provides a counterexample to a conjecture of Milterson from \cite{M95} that any algebraic data structure for polynomial evaluation over small fields that has space complexity $\poly(n)$ must have query complexity linear in $n$. We also note that a slightly more general version of \autoref{thm: data structure technical} is true where we have an appropriate tradeoff between the query and the space complexities. However, for the ease of exposition, we  focus on proving the specific statement in \autoref{thm: data structure technical}.

\subsection{Proof of \autoref{thm: data structure technical}}
The proof is a very simple application of the ideas in the multipoint evaluation algorithms discussed in the earlier section. The first ingredient is a reduction from the univariate problem to the multivariate problem. This step is also there in the data structure of Kedlaya \& Umans \cite{Kedlaya11}. 

\begin{definition}[Inverse Kronecker Map]
	\label{def:inverse-kronecker}
	Let $\F$ be a field. Then, for parameters $d, m \in \N$, the map $\psi_{d,m}$ from $\F[X]$ to $\F[Z_1, \ldots, Z_m]$ is defined as follows: Given a monomial $X^{t}$, write $t$ in base $d$, $t=\sum_{j\geq 0}t_jd^j$ and define the monomial $$M_a(\vecz):=Z_1^{t_0}Z_2^{t_1}\cdots Z_{m}^{t_{m-1}}.$$ 
	The map $\psi_{d,m}$ sends $X^a$ to $M_a(\vecz)$ and extends multilinearly to $\F[X]$.
\end{definition}
The map $\psi_{d,m}(f)$ can be computed in linear time in the size of $f$, assuming $f$ is represented explicitly by its coefficients. Also, $\psi_{d,m}$ is injective on the polynomials of degree less than $d^m$. For such polynomial $f$, if $F=\psi_{d,m}(f)$, then $$f(X)=F(X^{d^0}, X^{d^1},\ldots, X^{d^{m-1}}).$$

 Given a degree $n-1$ univariate polynomial $f$, let the parameters $m, d$ be set as follows:  $m = \log\log n$ and $d = n^{1/\log \log n}$ and construct the polynomial $F = \psi_{d, m}(f)$ on $m$ variables and degree at most $d-1$ in each variable. Clearly, this can be done in time $n\cdot \poly(\log n)$ by processing $f$ one monomial at a time. We now describe the preprocessing map and its image on $f$ denoted ${\cal D}_f$ using this polynomial $F$. This construction is based on the \autoref{thm:polynomial-evaluation-v1}. A slightly more general statement can be obtained by relying on the more involved algorithms for multipoint evaluation in \autoref{sec: mme algo large n} and \autoref{sec: mme large fields}, but for the proof of \autoref{thm: data structure technical}, Algorithm \autoref{algo:polynomial-evaluation-v1} is sufficient. 

\paragraph*{The preprocessing map.} 
Let $v_0(Y_0)$ be an irreducible polynomial over $\F_p$ of degree $a$ such that $\F_{q} =\F_p[Y_0]/\ideal{v_0(Y_0)}$, where $q = p^a$. We assume that this $v_0$ is given as a part of the input, since it conveys the description of the field we are working over. Compute the smallest integer of form $p^b$ such that $p^b > adm$. Compute an irreducible polynomial $v_1(Y_1)$ in $\F_q[Y_1]$ of degree $b$ such that $\F_{q^b}=\F_q[Y_1]/\ideal{v_1(Y_1)}$. Using \autoref{lem:subfield-construction}, compute the subfield $\F_{p^b}$ of  $\F_{q^b}$. We compute and store the evaluation of $F$  on every input in the grid $\F_{p^b}^m$. This is our ${\cal D}_f$. The space required to store all this data is at most $(m+1)p^{bm}$ elements of the field $\F_{p^{ab}}$ of equivalently $(m+1)p^{bm}b$ elements of the field $\F_{p^a}$. For our choice of parameters, the space complexity can be upper bounded as follows.  
\[
p^{bm} \leq (padm)^{m} = p^{\log \log n}\cdot (\poly(\log n))^{\log\log n} \cdot n \cdot {(\log \log n)}^{\log \log n} \leq n^{1 + o(1)}
\, .\]
Thus, the space complexity is at most $(m+1)p^{bm}b \leq n^{1 + o(1)}$ as claimed. Moreover, ${\cal D}_f$ can be computed by an algebraic algorithm over $\F_{q}$ using \autoref{lem: multidim FFT} and other ideas in \autoref{sec: mme algo large n}. 

\paragraph*{Answering evaluation queries using ${\cal D}_f$.}
We now describe an algorithm that given an $\alpha \in \F_q$ and access to ${\cal D}_f$ computes $f(\alpha)$ in time $n^{o(1)}$. The algorithm is essentially the same as the local computation step in Algorithm \autoref{algo:polynomial-evaluation-v1}. Note that ${\cal D}_f$ contains precisely the data that is computed in Algorithm \autoref{algo:polynomial-evaluation-v1} in the preprocessing phase. We assume the notation ($F, v_0, v_1$ etc.) set up in the previous paragraph. 
	\begin{algorithm}[H]
	\caption{ Evaluating polynomial $f(x)$ at a point in $\F_q$ using ${\cal D}_f$}
	\label{algo:polynomial-evaluation}
	\textbf{Input:} A point $\alpha\in\F_q$ and query access to ${\cal D}_f$.\\
	\textbf{Output:} $f(\alpha)$.\\
	\begin{algorithmic}[1]
		\State Let ${\pmb\alpha}=(\alpha^{d^0}, \alpha^{d^1},\ldots,\alpha^{d^{n-1}})$.
        \State Let ${\pmb\alpha}={\pmb\alpha}_{0}+{\pmb\alpha}_{1}Y_0+\cdots+{\pmb\alpha}_{a-1}Y_0^{a-1}$, where ${\pmb\alpha}_{j}\in \F_{p}^n$. 
		\State Compute ${\pmb\alpha}_{0},{\pmb\alpha}_{1},\ldots,{\pmb\alpha}_{k-1}$.
		\State Let $\vecg(t)$ be the curve defined as ${\pmb\alpha}_{0}+{\pmb\alpha}_{1}t+\cdots+{\pmb\alpha}_{a-1}t^{a-1}.$
		\State Compute the set $P=\{(\gamma,\vecg(\gamma)) \,\mid\, \gamma\in\F_{p^b}\}$.
		\State Collect the set $E=\{(\gamma, F(\pmb\gamma'))\,\mid\, (\gamma,\pmb\gamma')\in P\}$ by querying ${\cal D}_f$.
		\State Using $E$, interpolate the univariate polynomial $F(\vecg(t))$.
		\State Output $F(\vecg(Y_0))$ as $f(\alpha)$. 
	\end{algorithmic}
\end{algorithm}	
The proof of correctness of the construction of ${\cal D}_f$ and the query algorithm immediately follow from the proof of correctness of Algorithm \autoref{algo:polynomial-evaluation-v1}. The query complexity is clearly upper bounded by $p^b \leq padm$, which for our setting of parameters, i.e. $p = O(1)$, $a = \poly(\log n)$, $d = n^{1/\log \log n}$ and $m = \log \log n$ is $n^{o(1)}$. This completes the proof of \autoref{thm: data structure technical}.

\section{Rigidity uppper bounds}\label{sec:rigidity proof}

In this section, we prove \autoref{thm: vandermonde rigidity intro}.  We start with the definition of matrix rigidity which was  introduced by Valiant \cite{Valiant1977}.

\begin{definition}(Matrix rigidity)
For  a matrix $M$ over some field $\mathbb{F}$ and a natural number $r$, we define $R^{\mathbb{F}}_M(r)$ to be the smallest number $s$ for
which there exists a matrix $A$ with at most $s$ nonzero entries and a matrix $B$ of rank at most $r$ such that
$M = A+B$. If $R^{\mathbb{F}}_M(r) \geq s$, we say $M$ is $(r,s)$-rigid. When the underlying field is clear from the context, we drop the superscript and denote $R^{\mathbb{F}}_M(r)$ by $R_M(r)$.
\end{definition}

Now we define  \emph{regular rigidity} denoted by $r^{\mathbb{F}}_M(\cdot)$ .

\begin{definition}(Regular rigidity)
    For a matrix $M$ over some field $\mathbb{F}$ and a natural number $\tilde{r}$, we define $r_M(\tilde{r})$ to be the smallest number $s$
such that there exists a matrix $A$ with at most $s$ nonzero entries in each row and column and a matrix $B$ of rank at most $\tilde{r}$ such that
$M = A+B$. If $r^{\mathbb{F}}_M(\tilde{r}) \geq s$, we say $M$ is $(\tilde{r},s)$-regular rigid. Also here, when the underlying field is clear from the context, we drop the superscript and denote $r^{\mathbb{F}}_M(\tilde{r})$ by $r_M(\tilde{r})$.
\end{definition}

Note that the notion of \emph{regular rigidity} is weaker than the usual notion of  rigidity. That is, say $A$ is an $n \times n$ matrix and $A$ is $(r, ns)$-rigid then $r_A(r) \geq s$. From the perspective of rigidity upper bounds, if $r_A(r) \leq s$ then $R_A(r) \leq ns$.    Thus, it follows that proving a family of matrices to be non-regular rigid is even a stronger criterion compared to the usual notion of  non-rigidity. In this section, we  show that the Vandermonde matrices are not Valiant rigid (refer \autoref{sec: rigidity intro}). Note that, for any matrix $M\in \F^{n \times n}$, showing that there exists  constants $c_1, c_2$ and $\epsilon>0$ such that $$R_M^{\F}(\frac{n}{\exp({\epsilon^{c_1}}\log^{c_2}n)}) \leq n^{1+ \epsilon}$$ implies that $M$ is not Valiant rigid. For concreteness, we  state our results in terms of $R_M(\cdot)$ and $r_M(\cdot)$.

 We start by defining some interesting classes of matrices which will be used frequently in the subsequent sections.
 
\begin{definition}[Discrete Fourier Transform (DFT) Matrices]
Let $\F$ be any field and let $\omega\in\F$ be a primitive $n\mathrm{th}$ root of unity in $\F$. Then, the Discrete Fourier Transform (DFT) Matrix of order $n$, denoted by $F_n$ is an $n \times n$ matrix over $\F$ defined as follows. The rows and columns of the matrix are indexed from $\{0,1, \ldots, n-1\}$. The $(i,j)$th entry of $F_n$ is $\omega^{ij}$. 
\end{definition}
\begin{definition}[Circulant Matrices]\label{def: Circulant}
	Let $\F$ be a field and $C_n$ be a matrix of order $n \times n$ with entries in $\F$ such that the rows and columns of $C_n$ are indexed from $\{0,1, \ldots, n-1\}$. Then, $C_n$ is said to be a Circulant matrix if there exist $c_0,c_1, \ldots, c_{n-1} \in \F$ such that for every $i, j \in \{0,1, \ldots, n-1\}$, $C_n(i, j)= c_{{(i+j)}\mod n}$.
\end{definition}

\begin{definition}
	Let $\F$ be a field. A Vandermonde matrix$(V_n)$ of order $n \times n$ is defined by an input vector $\pmb\alpha = (\alpha_0,\alpha_1, \ldots, \alpha_{n-1})$ where $\pmb\alpha \in \F^n$. The rows and columns of $V_n$ are indexed from $\{0,1, \ldots, n-1\}$ and for every $i, j \in \{0,1, \ldots, n-1\}$, $V_n(i, j)= \alpha_i^{j}$.
\end{definition}

We now state the formal version of \autoref{thm: vandermonde rigidity intro}.
 
\begin{theorem}\label{thm: Vandermonde rigidity}
	Let $0 < \epsilon < 0.01$, $p$ be a fixed prime and $c$ be a fixed constant. Let $a: \N \to \N$ be a function such that for all  $n$, $a(n) \leq \log^c n$. Let $\{V_n\}_{n \in \N}$ be a family of matrices over $\overline{\F}_p$ such that for every $n$, $V_n$ is an $n \times n$ Vandermonde Matrix with entries in the subfield $\F_{q=p^{a(n)}}$. Then, for all large enough $n$,
	$$
	{R}^{\mathbb{F}_q}_{V_n} \Big{(}  \frac{n}{\exp( \Omega(\epsilon^7\log^{0.5} n)\big{)}} \Big{)} \leq   n^{1+ 31\epsilon} \, 
	$$
	
	
	
\end{theorem}

As alluded in  \autoref{sec: rigidity proof overview},  the main idea in proving \autoref{thm: Vandermonde rigidity} comes from viewing Algorithm \autoref{algo:polynomial-evaluation-v1} as a decomposition of (any) Vandermonde matrix into the product of a row-sparse matrix($A$) and  a sufficiently non-rigid matrix($B$). 
Before getting into the proof of the main theorem, we discuss the non-rigidity of the matrix family ``$B$". As it turns out, $B$ is a multidimensional analog of the DFT matrix. Thus,  we start by proving that the DFT matrices are non-rigid over quasi-polynomial size finite field $\F_q$ in \autoref{cor: DFT rigidity}. Followed by showing that the multidimensional  DFT matrices are also non-rigid in \autoref{thm: multi dim rigid upper bound}, using the non-rigidity of DFT matrices  and the Kronecker product structure of multidimensional DFT.  Finally, in \autoref{subsec: Van non-rigid}, we flesh out  the decomposition of Vandermonde matrices as product of row-sparse and multidimensional DFT matrices and use it to conclude that Vandermonde matrices are non-rigid as well. 

\subsection{Non-rigidity of DFT matrices}

In this section, we give an upper bound on the rigidity of DFT matrices over finite fields.  This bound follows directly from the work of Dvir and Liu \cite{DL20} although the precise statement needed for our purpose (see \autoref{cor: DFT rigidity}) is not included in \cite{DL20} as the results there are stated for a family of matrices over a fixed finite field. Whereas, in \autoref{thm: Vandermonde rigidity}, the field is also increasing in size with the dimension of the matrix and for its proof, we need a bound on the rigidity of a family of DFT matrices with the field growing in size with the dimension of the matrices. This turns out to be a direct consequence of the results in \cite{DL20} and was communicated to us by the Zeev Dvir and Allen Liu \cite{DLpersonal}. 

We start with a simple observation from \cite{DL20}. 
\vspace{-0.2cm}
\begin{observation}[\cite{DL20}]\label{obs: scaling property}
Let $M$ be an $n \times n$ matrix over a field $\F$, and let $D$ be an $n \times n$ diagonal matrix over $\F$. Then, for every choice of rank parameter $a$, 
$$ r_{DM}^{\F}(a)\leq r_{M}^{\F}(a) \, , \qquad \qquad and  \qquad \qquad r_{MD}^{\F}(a)\leq r_{M}^{\F}(a) \,. $$
\end{observation}
	

	
Let's analyse the DFT matrix $F_{q-1}$ over $\F_q$. We show that there is a way to get a Circulant matrix by scaling  $F_{q-1}$ matrix over $\F_q$. This along with  \autoref{obs: scaling property} links the rigidity of $F_{q-1}$ matrix over $\F_q$ with the rigidity of Circulant matrices.

	\begin{lemma} \label{clm: Scaling of Fourier and Circulant}
		It is possible to rescale the rows and columns of a DFT matrix $F_{q-1}$ over any finite field $\F_q$ to get a Circulant matrix $C_{q-1}$ over the extended field ${\mathbb{F}_{q^2}}$, where $q=p^a$. 
	\end{lemma}

	The proof of the above lemma is almost identical to the proof of Claim 2.22 in \cite{ DL20} and hence omitted. We would like to emphasise one key point though.  The scaling procedure involves working with an element $ \zeta$ s.t. $\zeta^2=g$, where $g$ is the generator of $\F^*_q$.  Note that $\zeta$ may not always exist in the underlying field $\mathbb{F}_q$, when $q$ is odd. In that case, we perform a degree $2$ extension over the base field $\mathbb{F}_q$.  After extension, we view $F_{q-1}$ over this extended field and perform appropriate scaling using $\zeta$ to get a  Circulant matrix over $\F_{q^2}$.  


We now state the upper bound on the rigidity of  Circulant matrices from \cite{DL20}. 
\begin{theorem}[Theorem $7.27$ in \cite{DL20}] \label{thm:Circulant rigidity}

Let $0 < \epsilon < 0.01$ and $p$ be a fixed prime. For all sufficiently large $n$, if $C_n$ is an $n \times n$ Circulant matrix over $\F_p$ then,
	
	$${r}^{\mathbb{F}_p}_{C_n} \Big{(} \frac{n}{\exp(\epsilon^6\big{(}\log n)^{0.35}\big{)}} \Big{)} \leq n^{15\epsilon}$$
	\end{theorem}
	
	We now state an upper bound on the rigidity of  DFT matrices due to Dvir \& Liu which is the main takeaway of this subsection. 

\begin{theorem}[\cite{DL20, DLpersonal}]\label{cor: DFT rigidity}
	Let $0 < \epsilon < 0.01$, $p$ be a fixed prime and $c$ be a fixed constant. Let $a: \N \to \N$ be a function such that for all  $n$, $a(n) \leq \log^c n$. Let $\{F_n\}_{n \in \N}$ be a family of matrices over $\overline{\F}_p$ such that for every $n$, $F_n$ is an $n \times n$ DFT Matrix with entries in the subfield $\F_{p^{a(n)}}=:\F_q$.
Then,
$${r}^{\mathbb{F}_q}_{F_n} \Big{(}  \frac{n}{\exp(0.5 \cdot \epsilon^6\big{(}\log n)^{0.35}\big{)}} \Big{)} \leq 2\log_p{q} \cdot n^{15\epsilon}$$
\end{theorem}

\begin{proof}

We first obtain a Circulant matrix $C_n$ of order $n$ by scaling the DFT matrix (of order $n$) from \autoref{clm: Scaling of Fourier and Circulant}. As discussed earlier, the entries in the corresponding Circulant matrix potentially belongs to a degree $2$ extension over the field $\mathbb{F}_q$, that is $C_n\in \F_{q^2}^{n \times n}$. 
Let $\mathbb{F}_{q^2}[X]=\mathbb{F}_p[X]/\langle v(X)\rangle$ where $v(X)$ is a degree $2a$ irreducible polynomial over $\mathbb{F}_p$. 
Hence,  $C_n$ can be expressed as  $ C_n=C^{(0)}+C^{(1)}X+C^{(2)}X^2+ \ldots +C^{(2a-1)}X^{2a-1}$, where each $C^{(i)}$ is also a Circulant matrix of order $n$ over $\mathbb{F}_p$. This  follows directly from the structure of  Circulant matrices.

Now we invoke the rigidity upper bound from \autoref{thm:Circulant rigidity} for each of the above Circulant matrix $C^{(i)}$ over the fixed field $\mathbb{F}_p$ . Due to the sub-additive property of rank and sparsity, the overall rank and sparsity of $C_n$ gets upper bounded by $2a$ times the rank and sparsity of each of these $C^{(i)}$'s. Hence,

	$${r}^{\mathbb{F}_q}_{F_n} \Big{(}  \frac{2an}{\exp(\epsilon^6\big{(}\log n)^{0.35}\big{)}} \Big{)} \leq 2an^{15\epsilon}$$
	
On rewriting the above equation, 
$${r}^{\mathbb{F}_q}_{F_n} \Big{(}  \frac{n}{\exp(\epsilon^6\big{(}\log n)^{0.35}- \log\log_pq^2\big{)}} \Big{)} \leq 2\log_pq \cdot n^{15\epsilon}$$

Since $a<  \log^c n $,  $\epsilon^6\log^{0.35} n \gg \log\log_pq^2$, this concludes the proof.

\end{proof}
\subsection {Non-rigidity of multidimensional DFT matrices}

For our proof of \autoref{thm: Vandermonde rigidity}, we need an upper bound on the rigidity of the following high dimensional analog of DFT Matrices. 
\begin{definition}
Let $m, d \in \N$, $\F$ be a field and $S \subset \F$ be a subset of  cardinality $d$. 
The matrix  $\mathcal{V}_{d,m}^{S}$ is a $d^m \times d^m$ matrix with  rows  labelled by all field elements of the product set $S^m = S \times S \times \cdots \times S$, and  columns  labelled by all monomials of individual degree at most $d-1$ in variables $\vx = (x_1, x_2, \ldots, x_m)$. Such that,  for $\va \in S^m$ and a monomial $\vx^{\ve}$ of individual degree at most $d-1$, the $(\va, \ve)$'th entry of $\mathcal{V}_{d,m}^{S}$ equals $\va^{\ve}$. 

\end{definition}

For most of this section, we  work with $\mathcal{V}_{d,m}^{S}$ for the setting where $\F$ is a finite field of size $q$,  $d = q$ and $S = \F$. Moreover, when $S$ is clear from the context, we drop the superscript and denote $\mathcal{V}_{d,m}^{S}$ by $\mathcal{V}_{d,m}$. 

Note that for $m = 1$ the matrix $\mathcal{V}_{q,1}^{\F_q}$ is closely related to the DFT matrix of order $q-1$ over $\F_q$. To see this, note that if $g$ is a generator of the multiplicative group $\F^*_{q}$, then $g$ is trivially a primitive root of unity of order $q-1$ over $\F_q$. So, up to a permutation of rows, the rows of $\mathcal{V}_{q,1}^{\F_q}$ can be viewed as being indexed by $g^0, g^1, g^2, \ldots, g^{q-1}$, and the columns by $0, 1, \ldots, q-1$ in this order. Thus, if we discard the row indexed by $0$ and the column indexed by $q-1$ of $\mathcal{V}_{q,1}^{\F_q}$, what remains is precisely a DFT matrix of order $q-1$ over the field $\F_q$. The following lemma is an easy consequence of this observation together with the upper bound on rigidity of DFT matrices from \cite{DL20}. 

\begin{lemma}\label{lem: univariate vand. rigidity}
	
	Let $0 < \epsilon < 0.01$,   and $\mathbb{F}_q$ be a finite field of size $q$. Let $V = \mathcal{V}_{q,1}^{\F_q}$. Then, for all sufficiently large $q $,

	$${r}^{\mathbb{F}_q}_{V} \Big{(} \frac{q}{\exp({\epsilon}^6 \big{(}\log q)^{0.34}\big{)}} \Big{)} \leq q^{16\epsilon}$$ 
	
\end{lemma}

\begin{proof}
	Consider the $(q-1)\times (q-1)$ submatrix $\bar{V}$ of $V$ by deleting the row corresponding to the field element $0$ and the column corresponding to the monomial of degree ${q-1}$. $\bar{V}$ is a DFT matrix of order $q-1$ with the generator of the multiplicative group $\mathbb{F}^*_q$ being the primitive root of unity of order $q-1$. Note that setting $n=q-1$ and invoking the rigidity upper bound of DFT matrices from \autoref{cor: DFT rigidity}, we get, 
	
		$${r}^{\mathbb{F}_q}_{\bar{V}} \Big{(}  \frac{(q-1)}{\exp(0.5 \epsilon^6\big{(}\log (q-1))^{0.35}\big{)}} \Big{)} \leq 2\log_pq \cdot (q-1)^{15\epsilon}.$$
		
		Before simplifying the above expression, let's see how this related to rigidity of  ${V}$. In order to do that, we re-add the deleted row and column to $\bar{V}$. This process can potentially increase the rank of the augmented matrix by at most $2$. Also, note that $\log_pq \ll q^{\epsilon} $, thus we get

		

		
			
			
			$${r}^{\mathbb{F}_q}_{{V}} \Big{(}  \frac{q}{\exp(\epsilon^6\big{(}(\log q)^{0.34}\big{)}} \Big{)} \leq q^{16\epsilon}.$$

\end{proof}
In the rest of this section, we generalize \autoref{lem: univariate vand. rigidity} to 
obtain an upper bound on the rigidity of $\mathcal{V}_{q,m}^{\F_q}$ for larger values of $m$. The following simple observation is the first step in this direction.  
\begin{observation}{\label{obs:Kronecker product uni}}
$
\mathcal{V}_{q, m}^{\F_q} =\footnote{This equality holds up to some row column permutation, which doesn't affect rigidity.} \left(\mathcal{V}_{q, 1}^{\F_q}\right)^{\otimes m} \,. 
$
\end{observation}
We are now ready to state the main theorem of this section, where we prove a non-trivial upper bound on the rigidity of $\mathcal{V}_{q, m}^{\F_q}$ for an appropriate range of parameters. 
	

\begin{theorem}\label{thm: multi dim rigid upper bound}
Let $0 < \epsilon < 0.01$, and $m \in \mathbb{N^+}$  be the given parameters and $\mathbb{F}_q$ be any finite field such that $m \leq q^\epsilon$. Then for large enough $q$  , we have 
		$${r}^{\mathbb{F}_q}_{\mathcal{V}_{q, m}} \Big{(}\frac{q^m}{\exp{{({{{9  \epsilon^7}}} {(}\log q)^{0.34}m)}}}\Big{)} \leq q^{27\epsilon \cdot m}$$
	\end{theorem}

\begin{proof}
	From \autoref{obs:Kronecker product uni}, we have 
	$$
\mathcal{V}_{q, m}^{\F_q} = \left(\mathcal{V}_{q, 1}^{\F_q}\right)^{\otimes m} \, . 
$$
From \autoref{lem: univariate vand. rigidity}, we have that 
	$${r}^{\mathbb{F}_q}_{\mathcal{V}_{q,1}} \Big{(} \frac{q}{\exp( {{\epsilon}^6} \big{(}\log q)^{0.34}\big{)}} \Big{)} \leq q^{16\epsilon} \, .
	$$	
Thus, $\mathcal{V}_{q,1}$ can be written as the sum of a matrix $L$ of rank at most $\frac{q}{\exp({{\epsilon}^6} \big{(}\log q)^{0.34}\big{)}}$ and a matrix $S$ with row sparsity at most $q^{16\epsilon}$. Thus,  $\mathcal{V}_{q, m}^{\F_q}$ can be written as 
$$
\mathcal{V}_{q, m}^{\F_q} = (L_1+S_1) \otimes (L_2+S_2) \otimes \ldots \otimes (L_m+S_m) \, ,
$$
where each $L_i$ is a copy of $L$ and each $S_i$ is a copy of $S$ (they have been indexed for clarity	of notation). 

The above Kronecker product has $2^m$ many terms, each term consisting of the Kronecker products of various copies of $L_i$ and $S_j$. We  partition these summands into two groups based on the number of copies of $L$ participating in the Kronecker product. To complete the proof, we show that the sum of every term with many copies of $L$ is a matrix of \emph{not too high} rank and the sum of the remaining terms (that have few copies of $L$ and hence many copies of $S$) are non-trivially sparse. This would complete the proof of the theorem. We now fill in the details, which involve some slightly careful calculations to get the quantitative bounds stated in the theorem.

We pick a threshold $t$ (to be set later) and collect all terms in the Kronecker product 	$$
\mathcal{V}_{q, m}^{\F_q} = (L_1+S_1) \otimes (L_2+S_2) \otimes \ldots \otimes (L_m+S_m) \, ,
$$
consisting of at most $t$ copies of $S$ for $t = m(1-10\epsilon)$. Let  $\mathcal{L}$ be the sum of all such terms, and let $\mathcal{S}$ be the sum of the remaining terms. In the following two claims, we obtain an upper bound on the rank of $\mathcal{L}$ and the sparsity of $\mathcal{S}$. 
\begin{claim}\label{clm: rank upper bound}
\[
\rank(\mathcal{L}) \leq \frac{q^m}{\exp(9 \epsilon^7(m\cdot (\log q)^{0.34}))} \, .
\]
%
\end{claim}
\begin{claim}\label{clm: sparsity upper bound}
Every row and column of $\mathcal{S}$ has at most $q^{27\epsilon \cdot m}$ non-zero entries.  
\end{claim}
These bounds, together with the decomposition 
\[
\mathcal{V}_{q,m}^{\F_q} = \mathcal{L} + \mathcal{S} \, 
\]
complete the proof of the theorem. 
\end{proof}

We now prove \autoref{clm: rank upper bound} and \autoref{clm: sparsity upper bound}. 

\begin{proof}[Proof of \autoref{clm: rank upper bound}]
From \autoref{lem: univariate vand. rigidity}, we know the rank of each $L_i$ is at most $\Big{(} \frac{q}{\exp \big{(}{{\epsilon}^6} {(}\log q)^{0.34}\big{)}} \Big{)}$ and the rank of each $S_i$ is at most $q$ since it is a $q \times q$ matrix. Moreover, the rank of a Kronecker product of matrices is the product of their ranks. 
Thus, we have 
	\begin{align*}
	\rank(\mathcal{L}) &\leq \sum_{j=0}^{t-1} {{m} \choose {j} } {\Big{(}\frac{q}{2^\alpha}\Big{)}}^{m-j} q^j \quad \quad \quad \quad \quad \quad \text{(where $\alpha=  {\epsilon}^6 \big{(}\log q)^{0.34}$)} \\
	&= {\Big({}\frac {q}{2^\alpha}}\Big{)}^m \sum_{j=0}^{t-1}{{m} \choose {j} } 2^{\alpha j} \\
	 &\leq {\Big({}\frac {q}{2^\alpha}}\Big{)}^m  2^{\alpha t}  \sum_{j=0}^{t-1}{{m} \choose {j} } \\
	 & \leq {\Big({}\frac {q}{2^\alpha}}\Big{)}^m  2^{\alpha t} \cdot 2^m \quad \quad \quad \quad \quad \quad \quad  \\
	&=  {({} {q}}{)}^m   2^{\alpha(t-m) + m} \, .
	\end{align*}
We now obtain an upper bound on the quantity $ 2^{\alpha(t-m)+ m}$ to complete the proof of the claim.
\begin{align*}
2^m 2^{\alpha(t-m)}  
&\leq 2^{(-10 \epsilon \alpha + 1 )m} \quad \quad \quad \quad \quad \text{(since $t = m(1-10\epsilon $))}\\
&\leq 2^{-\alpha \cdot 9 \epsilon m} \quad \quad \quad \quad \quad \text{since $\alpha \cdot 9 \epsilon > 1$ for all large enough $q$}\\
\end{align*}
Thus, the rank of $\mathcal{L}$ is at most ${q^m}\cdot {2^{-9 \epsilon \alpha m}}$, which by the choice of $\alpha$ gives 
\[
\rank(\mathcal{L}) \leq \frac{q^m}{\exp( 9 \epsilon^7 m\cdot (\log q)^{0.34})} \, .
\]

\end{proof}
\begin{proof}[Proof of \autoref{clm: sparsity upper bound}]
This proof also proceeds along the lines of the proof of \autoref{clm: rank upper bound}, and involves similar calculations. We just note that every summand in $\mathcal{S}$ involves at least $(m - t)$ sparse matrices in the Kronecker product. Moreover, we rely on the basic fact that the row/column  sparsity of a Kronecker product of matrices is equal to the product of the row/column sparsity of each of the matrices in the product. 

Let $q'= q^{16\epsilon}$ denote the row sparsity of every matrix $S_i$. For each $L_i$, we use the obvious upper bound of $q$ on its row sparsity for our estimate. 
\begin{align*}
\rspars(\mathcal S) &\leq \sum_{j=t}^{m} {{m} \choose {j} }(q')^j {(q)}^{m-j} \\
&= q^m \cdot \sum_{j=t}^m {{m} \choose {j}} \Big{(}{\frac{q'}{q}} \Big{)}^j \\
&\leq q^m \sum_{j=t}^m m^j \Big{(}{\frac{q'}{q}} \Big{)}^j \\
&\leq 2.q^m  \Big{(}{\frac{q'm}{q}} \Big{)}^t \quad \quad \quad \quad \quad \quad \text{(since $q'm/q < 1/2$)} \, . \\
&\leq q^m  \Big{(}{\frac{1}{q^{1-16 \epsilon- \log_q m}}} \Big{)}^t \,  \\
&\leq q^m  \Big{(}{\frac{1}{q^{1-17 \epsilon}}} \Big{)}^t \quad \quad \quad \quad \quad \quad \text{(we have $\frac{\log m}{\log q} < \epsilon$)} \, 
\end{align*} 

Using $t = m(1-10\epsilon)$, we get 
\begin{equation*}
\rspars(\mathcal S) \leq  q^{m(1-(1-10 \epsilon)(1-17 \epsilon))}  \leq q^{(10 \epsilon+ 17 \epsilon)m} \leq q^{27 \epsilon \cdot m}  \,. 
\end{equation*}

An almost identical argument also bounds the column sparsity, which completes the proof of the claim. 

\end{proof}

\subsection{Non-rigidity of Vandermonde matrices}\label{subsec: Van non-rigid}
In this section, we  prove \autoref{thm: Vandermonde rigidity}. 
As discussed in \autoref{sec: rigidity proof overview}, using the algorithm for multivariate multipoint evaluation, we  write (any) Vandermonde matrix as the product of a sparse matrix and a multidimensional DFT matrix. Then, we invoke the rigidity upper bound for multidimensional DFT matrices in \autoref{thm: multi dim rigid upper bound}. This together with the fact that a product of sparse and non-rigid matrix continues to be non-rigid with slightly diminished parameters completes the proof. We now fill in the details. 
\begin{proof}[Proof of \autoref{thm: Vandermonde rigidity}]
We start by setting up some necessary notation. 
Let $d, m$ be parameters, chosen as follows. 
\begin{itemize}
\item $m :=  \ceil{\log^{0.3}n}$
\item  $d$ is set to be the smallest integer such that $d^m > n$
\end{itemize}
   Clearly,  $d$ can be taken to be at most $n^{1/m} + 1$. For this choice of $d, m$, we note that $d^m$ is at least $n$ and most $(n^{1/m} + 1)^m \leq 2n$. 
   
   Let $\alpha_0, \alpha_1, \ldots, \alpha_{n-1}$ be the generators of the Vandermonde matrix $V_n$
, i.e., $V_n(i, j) = \alpha_i^j$. For every $i \in \{0,1, \ldots, n-1\}$, let $\pmb\alpha_i = (\alpha_i, \alpha_i^d, \ldots, \alpha_i^{d^{m-1}})$. For simplicity, we use $a$ to denote $a(n)$. Let $b$ be the smallest integer such that $padm \geq q_0 = p^b > adm$. Let $W$ be the $p^{bm} \times p^{bm}$ matrix with its rows indexed by vectors in $\F_{p^b}^m$ and the columns indexed by all $m$-variate monomials of individual degree at most $q_0 - 1$ and the $(\vecc, \vece)$ entry of $W$ is equal to $\vecc^{\vece}$. Moreover, let $\tilde{V}$ be the matrix with rows indexed by $\{0, 1, \ldots, n-1\}$ and columns indexed by all $m$-variate monomials of individual degree $d-1$ and the $(i, \vece)$ entry being equal to $\pmb\alpha_i^{\vece}$. 

To see the connection between the matrices $V_n$, $\Tilde{V}$ and $W$, let us try to understand the action of $V_n$ on a vector. Semantically, we can view this $n$ dimensional vector as the coefficient vector of a univariate polynomial $f$  of degree at most $n-1$ and thus the matrix vector product $V_n\cdot \coeff(f)$ is precisely the evaluation vector of $f$ on inputs $\alpha_0, \alpha_1, \ldots, \alpha_{n-1}$. Recall the inverse Kronecker map $\psi_{d,m}$ from \autoref{def:inverse-kronecker}, and let $F$ be an $m$-variate polynomial of degree at most $d-1$ in each variable such that $F = \psi_{d,m}(f)$. For this to make sense, recall that by our choice of parameters $d^m \geq n$. Now, it  follows from these definitions that for every $i \in \{0, 1, \ldots, n-1\}$, $f(\alpha_i) = F(\pmb\alpha_i)$. Moreover, the coefficient vectors of $f$ and $F$ are closely related. In fact, if $d^m = n$, then these coefficient vectors are exactly the same (even though the natural labelling of the coordinates of $\coeff(f)$ is via univariate monomials of degree at most $n-1$ and that of $\coeff(F)$ is via $m$-variate monomials of individual degree at most $d-1$). If $d^m > n$, then we have to append some zeroes to the coefficient vector $f$ to obtain the coefficient vector of $F$. In other words, 
\[
\coeff(F) = \Tilde{I} \cdot \coeff(f) \, ,
\]
where $\Tilde{I}$ is a $d^m \times n$ matrix with the top $n \times n$ submatrix being the identity matrix and the remaining $d^m - n$ rows being all zeroes. Moreover, 
\[
\Tilde{V} \cdot \coeff(F) = V_n \cdot \coeff(f)\, .
\]


Now we recall Algorithm \autoref{algo:polynomial-evaluation-v1} for multivariate multipoint evaluation and invoke it for the polynomial $F$ and evaluation points $\pmb\alpha_0, \pmb\alpha_1, \ldots, \pmb\alpha_{n-1}$. The algorithm first evaluates $F$ on the product set $\F_{q_0}^m$ in the preprocessing phase. In other words, it computes the vector $W\cdot \coeff(F)$. Then, for every $i \in \{0, 1, \ldots, n-1\}$, the value of $F$ on $\pmb\alpha_i$ is computed by some \emph{local} computation consisting of looking at the univariate polynomial $h_i(t)$ obtained as a restriction of $F$ on a curve of degree at most $a-1$ through $\pmb\alpha_i$, then interpolating $h_i$ using the already available values of $F$ on $\F_{q_0}^m$, and then evaluating $h_i$ on an appropriate input to get $F(\pmb\alpha_i)$. In other words, $F(\pmb\alpha_i)$ is obtained by taking an appropriate weighted linear combination of the value of $F$ on a specific subset of points in $\F_{q_0}^m$ of size at most $adm$. Thus, for every $i \in \{0, 1, \ldots, n-1\}$, there is a vector $\tau_i$ of length $q_0^m$ with entries in the field $\F_{p^{ab}}$\footnote{Recall that for technical reasons, we have to work in the field $\F_{p^{ab}}$ in Algorithm \autoref{algo:polynomial-evaluation-v1} even though the inputs are in $\F_{p^a}$. } such that
\[
F(\pmb\alpha_i) = \ideal{\tau_i, W\cdot \coeff(F)} \, .
\]
If we collect the vectors $\tau_0, \tau_1, \ldots, \tau_{n-1}$ into an $n \times q_0^m$ matrix $\Gamma$, then we have 
\[
\Tilde{V} \cdot \coeff(F) = \Gamma \cdot W \cdot \coeff(F) \, .
\]
Recalling the relation between $F$, $f$ and between $V_n$, $\Tilde{V}$, we get that 
\[
V_n \cdot \coeff(f) = \Tilde{V} \cdot \coeff(F) = \Gamma \cdot W \cdot \coeff(F) \, ,
\]
or, using  $\coeff(F) = \Tilde{I} \cdot \coeff(f)$, we get 
\[
V_n \cdot \coeff(f) = \left( \Gamma \cdot W \cdot \tilde{I}\right) \cdot \coeff(f) \, .
\]
Recall that we started with $\coeff(f)$ being an arbitrary vector. Thus, we have 
\[
V_n = \Gamma \cdot W \cdot \tilde{I} \, .
\]
Now, to obtain a decomposition of $V_n$ as the sum of a sparse and a low rank matrix, the idea is to invoke \autoref{thm: multi dim rigid upper bound} on $W$, and combine the decomposition obtained together with the sparsity of $\Gamma$ and $\Tilde{I}$ to obtain a similar decomposition for $V_n$.  To invoke \autoref{thm: multi dim rigid upper bound}, we must satisfy all the constraints on the parameters present there. To this end, we note that
\begin{itemize}
    \item $q_0 > adm$ and thus is large enough.
    
    \item $\log m / \log q_0 < \epsilon$. This is true since  $\log m / \log q_0 < \log m / \log d < \frac{0.3 \log\log n}{\log^{0.7} n } \ll \epsilon$. 
\end{itemize}
Now from \autoref{thm: multi dim rigid upper bound}, we have that $W$ can be written as the sum of matrices $L$ and $S$ where   the row and column sparsity of $S$ is at most $q_0^{27\epsilon m}$ and the rank of $L$ can be upper bounded as follows.
\begin{align*}
\rank(L) \leq    \frac{{q^m_0}}{\exp(9\epsilon^7(\log q_0)^{0.34}m)} &\leq \frac{(padm)^m}{\exp(9\epsilon^7(\log q_0)^{0.34}m)}  \, . 
\end{align*}
Note that by our choice of parameters, $(pam)^m \ll 2^{\log^{0.4}n}$ and ${m\log^{0.34} q_0} \geq (\log n)^{{(0.34)}{(0.70)}+0.3} \geq {\log^{(0.54)} n}$. Also, since $2^{(\epsilon^7m\log^{0.34} q_0)} \gg (pam)^m$, we get
\[
\rank(L) \leq \frac{{q^m_0}}{\exp(9\epsilon^7(\log q_0)^{0.34}m)} \leq \frac{n}{\exp( \Omega(\epsilon^7\log^{0.5} n)\big{)}} \, .
\]
Now, we have
\[
V_n = \Gamma \cdot W \cdot \Tilde{I} = \Gamma \cdot (L + S) \cdot \Tilde{I} \, ,
\]
or after simplification, 
\[
V_n = (\Gamma \cdot L \cdot \Tilde{I}) + (\Gamma \cdot S \cdot \Tilde{I})\, .
\]
Clearly, the rank of $(\Gamma \cdot L \cdot \Tilde{I})$ is at most the rank of $L$ which as we calculated above is at most $\frac{n}{\exp( \Omega(\epsilon^7\log^{0.5} n)\big{)}}$. The row sparsity of $(\Gamma \cdot S \cdot \Tilde{I})$ is at most the product of the row sparsity of $\Gamma$ and the row sparsity of $S$ \footnote{$\Tilde{I}$ has row and column sparsity $1$, so it does not affect the calculations.} and thus is at most 
\begin{align*}
(adm)(q^{27\epsilon \cdot m}_0) &\leq (padm) ((q_0)^m)^{27\epsilon} \\
&= p^{27\epsilon\cdot m+1} \cdot a^{27\epsilon\cdot m+1} n^{27\epsilon} m^{(27\epsilon \cdot m+1)}d
\end{align*}
Note that by our choice of parameters, $d < n^{\epsilon}$, $a^{27\epsilon\cdot m+1} < n^\epsilon$, $p^{27\epsilon\cdot m+1} < n^\epsilon$  and $m^{(27\epsilon \cdot m+1)} < n^{\epsilon}$, and thus, the row sparsity of $(\Gamma \cdot S \cdot \Tilde{I})$ is at most $n^{31 \epsilon}$. 

Thus, $V_n$ can be written as the sum of a matrix of row sparsity\footnote{We emphasize the point that we only deal with row sparsity here and \emph{not both row and column sparsity}.} at most $n^{31 \epsilon}$ and rank at most $\frac{n}{\exp( \Omega(\epsilon^7\log^{0.5} n)\big{)}}$ as claimed in the theorem.

\end{proof}

\section*{Acknowledgements}
Mrinal thanks Swastik Kopparty for introducing him to the work of Kedlaya \& Umans \cite{Kedlaya11} and the question of multipoint evaluation and numerous invaluable discussions.

We thank Zeev Dvir and Allen Liu for answering our questions regarding the results in \cite{DL20} and for allowing us to include the proof sketch of \autoref{cor: DFT rigidity} in this paper. We also thank Ben Lund for helpful discussions and references on the finite fields Kakeya problem and for pointing us to the relevant literature on Furstenberg sets, both of which indirectly played a role in some of the ideas in this paper.

Finally, we thank Prerona Chatterjee, Prahladh Harsha, Ramprasad Saptharishi and Aparna Shankar for sitting through a (pretty sketchy) presentation of  earlier versions of some of the proofs in this paper and for much  encouragement.

{\small
\bibliographystyle{prahladhurl}
\bibliography{jrnl-names-abb,prahladhbib,crossref,references}
}

\end{document}

%% file: thmmacros.tex
\usepackage{amsthm}
\usepackage{thmtools,thm-restate}

\numberwithin{equation}{section}
\declaretheoremstyle[bodyfont=\it,qed=\qedsymbol]{noproofstyle}

\declaretheorem[numberlike=equation]{observation}

\declaretheorem[name=Observation,numbered=no]{observation*}

\declaretheorem[numberlike=equation]{theorem}

\declaretheorem[name=Theorem,numbered=no]{theorem*}

\declaretheorem[numberlike=equation]{lemma}
\declaretheorem[name=Lemma,numbered=no]{lemma*}

\declaretheorem[numberlike=equation]{corollary}
\declaretheorem[name=Corollary,numbered=no]{corollary*}

\declaretheorem[numberlike=equation]{proposition}
\declaretheorem[name=Proposition,numbered=no]{proposition*}

\declaretheorem[numberlike=equation]{claim}
\declaretheorem[name=Claim,numbered=no]{claim*}

\declaretheorem[name=Conjecture,numbered=no]{conjecture*}

\declaretheorem[numberlike=equation]{question}
\declaretheorem[name=Question,numbered=no]{question*}

\declaretheoremstyle[bodyfont=\it,qed=$\lrcorner$]{defstyle} 

\declaretheorem[numberlike=equation,style=defstyle]{definition}
\declaretheorem[unnumbered,name=Definition,style=defstyle]{definition*}

\declaretheorem[unnumbered,name=Example,style=defstyle]{example*}

\declaretheorem[unnumbered,name=Notation=defstyle]{notation*}

\declaretheorem[unnumbered,name=Construction,style=defstyle]{construction*}

\declaretheorem[numberlike=equation,style=defstyle]{remark}
\declaretheorem[unnumbered,name=Remark,style=defstyle]{remark*}


%% file: bibmacros.tex
\usepackage{nth}
\usepackage{intcalc}
\usepackage{etoolbox}
\usepackage{xstring}

\usepackage{ifpdf}
\ifpdf
\else
\usepackage[quadpoints=false]{hypdvips}
\fi

\newcommand{\ECCCurl}[2]{http://eccc.hpi-web.de/report/\ifnumcomp{#1}{>}{93}{19}{20}#1/#2/}

\newcommand{\shortECCC}[2]{\texttt{\href{http://eccc.hpi-web.de/report/\ifnumcomp{#1}{>}{93}{19}{20}#1/#2/}{eccc:TR#1-#2}}}

\newcommand{\parseECCC}[1]{
\StrSubstitute{#1}{TR}{}[\tmpstring]%
\IfSubStr{\tmpstring}{/}{ 
\StrBefore{\tmpstring}{/}[\ecccyear]%
\StrBehind{\tmpstring}{/}[\ecccreport]%
}{
\StrBefore{\tmpstring}{-}[\ecccyear]%
\StrBehind{\tmpstring}{-}[\ecccreport]%
}%
\shortECCC{\ecccyear}{\ecccreport}}